\documentclass{IEEEtran}
\usepackage{amssymb}
\usepackage{amsfonts}
\usepackage{amsthm}
\usepackage{amsmath}
\usepackage{subcaption}
\usepackage {graphicx}
\usepackage{tikz}
\usepackage{pgfplots}
\usepackage{algorithm}
\usepackage{algpseudocode}
\usepackage{relsize}
\def\a{\mathbf{a}}
\def\b{\mathbf{b}}
\def\c{\mathbf{c}}

\def\e{\mathbf{e}}
\def\g{\mathbf{g}}

\def\o{\mathbf{1}}
\def\p{\mathbf{p}}

\def\u{\mathbf{u}}
\def\v{\mathbf{v}}

\def\x{\mathbf{x}}
\def\y{\mathbf{y}}
\def\z{\mathbf{z}}
\def\xhat{\hat{\mathbf{x}}}
\def\yhat{\hat{\mathbf{y}}}

\def\xhats{\hat{x}}

\def\A{\mathbf{A}}
\def\C{\mathbf{C}}
\def\D{\mathbf{D}}

\def\G{\mathbf{G}}

\def\I{\mathbf{I}}
\def\P{\mathbf{P}}
\def\Q{\mathbf{Q}}
\def\R{\mathbf{R}}
\def\T{\mathbf{T}}

\def\diag{\mathbb{D}}
\def\mean{\mathbb{E}}

\def\pr{\mathbb{P}}

\newtheorem{thm}{Theorem}
\newtheorem{lem}{Lemma}
\newcommand{\SUM}[4]{%
{\displaystyle \sum_{{#1}={#2}}^{#3}\!{\!#4}
}%
}
\newcommand{\PROD}[4]{%
{\displaystyle \prod_{{#1}={#2}}^{#3}\!{\!#4}
}%
}
\newcommand{\INT}[3]{%
{\displaystyle \int\limits_{#1}^{#2}\!{\!#3}
}%
}
\newcommand{\NrmOne}[1]{\|{#1}\|_1}
\newcommand{\NrmTwo}[1]{\|{#1}\|_2}

\DeclareMathOperator*{\argmin}{arg\,min}
\def\rk[#1]{{\color{red}\sffamily\small\em $\Rightarrow$ #1 $\Leftarrow$}}
\title{Detection Estimation and Grid matching of Multiple Targets with Single Snapshot Measurements}
\author{Rakshith Jagannath
\thanks{The author is with the Department of Electrical Engineering, I.I.T Madras, Chennai, India (e-mail:ee13d005@ee.iitm.ac.in)}}
\begin{document}
\maketitle
\begin{abstract}
In this work, we explore the problems of detecting the number of narrow-band, far-field targets and estimating their corresponding directions from single snapshot measurements. The principles of sparse signal recovery (SSR) are used for the single snapshot detection and estimation of multiple targets. In the SSR framework, the DoA estimation problem is grid based and can be posed as the lasso optimization problem. However, the SSR framework for DoA estimation gives rise to the grid mismatch problem, when the unknown targets (sources) are not matched with the estimation grid chosen for the construction of the array steering matrix at the receiver. The block sparse recovery framework is known to mitigate the grid mismatch problem by jointly estimating the targets and their corresponding offsets from the estimation grid using the group lasso estimator. The corresponding detection problem reduces to estimating the optimal regularization parameter ($\tau$) of the lasso (in case of perfect grid-matching) or group-lasso estimation problem for achieving the required probability of correct detection ($P_c$). We propose asymptotic  and finite sample test statistics for detecting the number of sources with the required $P_c$ at moderate to high signal to noise ratios. Once the number of sources are detected, or equivalently the optimal $\hat{\tau}$ is estimated, the corresponding estimation and grid matching of the DoAs can be performed by solving the lasso or group-lasso problem at $\hat{\tau}$. 
\end{abstract}
\section{Introduction}
\label{intro}
Detection, estimation and tracking of targets are the primary functions of radar-based localization systems. A main challenge frequently faced by these systems is the problem of restricted measurements due to limited availability of sensors. In such cases, it is essential to exploit the sparsity of targets in the array manifold (spatial domain) for the purpose of detection and estimation with as few sensors as possible. In this work, we focus on the problems of detecting the number of narrow-band, far-field targets and estimating their corresponding direction of arrivals (DoAs) from single snapshot measurements. 

The signal model used for detection and estimation in single snapshot DoA problem models the observed measurements as a continuous and non-linear function of the DoAs \cite{van2002detection}. As the DoAs are sparse in the spatial domain, sparse signal recovery (SSR) based techniques can be used for detection and estimation. In the SSR framework, the continuous DoA signal model can be approximated into three classes, namely, on-grid, off-grid and grid-less \cite{four}. 

In the on-grid SSR framework, the signal model for estimation is obtained by the discretization of the continuous DoAs over a selected interval to construct the array steering matrix over an estimation grid of DoAs. The true DOA targets are then assumed to lie on the estimation grid and SSR based estimators have been proposed for DoA estimation. These estimators essentially use the lasso estimator in its various forms for estimation of the DoAs \cite{complexlars}. However, the lasso regularization parameter ($\tau$), which controls the number of sources that are estimated is usually chosen empirically. In the case of sparse greedy algorithms like orthogonal matching pursuit (OMP) and its variants, the number of sources is assumed to be known apriori and then the estimation is performed \cite{arx}. 

For the case of a single source in noise model in the on-grid SSR framework, the estimate of the regularization parameter, $\hat{\tau} = \sigma\sqrt{-\ln(P_f)}$ for a given probability of false alarm $P_f$ and noise variance $\sigma$, was obtained in \cite{fuchspaper} using the generalized likelihood ratio test (GLRT). However, for multiple targets, it is well-known that the GLRT selects the largest model \cite{kaydetection}. Algorithms based on cross-validation and information criteria principles like Bayesian information criteria and minimum description length have been proposed in \cite{cs_cv,Mosesequivalence,mdl,OSMDL}. But, these algorithms are known to suffer in detection performance for small number of snapshots and are mostly not even applicable for the single snapshot case \cite{buckley}. Also, the relationship between $\tau$ and the probability of correct detection, $P_c$ (or $P_f$) have not been obtained for these algorithms. A number of asymptotic results (for large measurements) which are the SSR counterparts to the martingale stability theorem \cite{mst} derived for maximum likelihood estimation framework exist in the literature \cite{chenbp,bnbut}, wherein the optimal regularization parameter ($\hat{\tau}$) is derived to minimize the lasso estimation error. But, small estimation errors does not necessarily mean that sparsity and support of the estimate is same as the original parameter, which is required to control $P_c$ (or $P_f$) in the detection framework. In the related framework of sequential hypothesis testing, family-wise error rate control procedures and the Benjamini-Hojberg procedure and its variants have been used for controlling the false discovery rates and $p$-values (these quantities can be related to $P_c$). However, to our knowledge, most of the results are asymptotic in measurements and offer average rate control with respect to (w.r.t) $p$-values for large measurements. Hence these are useful mostly for the multiple snapshot DoA detection and estimation. In \cite{siglass}, the co-variance test statistics has been proposed for real measurements to obtain the optimal $\tau$. However, the authors obtain an asymptotic (in  the number of measurements) distribution for the co-variance test statistics, which can then be used to obtain the optimal $\tau$ for an approximate $P_c$. 

The on-grid SSR framework for DoA estimation gives rise to the grid mismatch problem when the unknown targets (sources) do not lie in the estimation grid, chosen for the construction of the array steering matrix at the receiver. The block sparse recovery framework is known to mitigate the grid mismatch problem by jointly estimating the targets and their corresponding off-sets from the estimation grid using the group-lasso estimator or its variants. The corresponding detection problem reduces to estimating the optimal regularization parameter ($\tau$) of the group-lasso estimation problem for achieving the required probability of correct detection ($P_c$). A number of estimation algorithms have been proposed for joint DoA estimation and grid matching in the block sparse recovery framework using second order cone programming, semi-definite programming and greedy algorithms \cite{elsTLS,gengmm,Teke-2014}. But to our knowledge, the problem of detection of the number of sources has not been sufficiently explored.

The grid-less methods for DoA estimation such as MUSIC and ESPRIT traditionally require the knowledge of the number of sources for estimation of DoAs and an estimate of the measurement co-variance matrix, which in-turn requires multiple snapshots. Hence, these cannot be used for detection and estimation of DoAs with single snapshot measurements. Recently, other sub-space based algorithms for single snapshot DoA estimation have been proposed in \cite{four,l1svd,ssmusic}, but they all require the knowledge of the number of sources and hence do not detect the number of sources from the measurements. Since we work with single snapshot measurements, beam-formers can be used only for detecting a single source, but these techniques cannot be used for detecting multiple sources with adequate performance \cite{ssdoar}. 

In this work, we explore the problem of finding the relationship between $\tau$ and the detection performance metrics like the probability of correct detection ($P_c$), the probability of mis-detection ($P_m$) and the probability of false alarm ($P_f$). Specifically, we propose finite sample and asymptotic test statistics which can be used at moderate to high SNRs to obtain the optimal $\tau$ for a given $P_c$ with varying degrees of performance. This is accomplished by comparing the test statistics to a threshold which is obtained by inverting the cumulative distribution function (c.d.f) of the proposed test statistics. Finally, we compare the performance of all these tests through simulations and discuss their merits. 

\emph{Organization and Notations:} We use bold lower case letters to denote vectors ($\x$) and bold upper case letters to denote matrices ($\A$). $\|\x\|_{\infty}$, $\|\mathbf{x}\|_1$ and $\|\mathbf{x}\|_2$ denote the $l_{\mathsmaller{\infty}}$, $l_1$ and $l_2$ norms of a vector $\mathbf{x}$ respectively. $\mathbf{x}^{H}$ denotes the Hermitian of $\mathbf{x}$. $\diag(\x)$ denotes a diagonal matrix with entries of $\x$ as the diagonal elements, $\pr(.)$ denotes probability and $\mean(.)$ denotes expectation. $\x\odot\y$ represents the Hadamard product (entry-wise product) of two vectors $\x$ and $\y$. The rest of this work is organized as follows. Section-\ref{model} describes the DoA signal model used in this work. Section-\ref{main_content} describes the main content, wherein we propose algorithms for joint detection, estimation and grid matching of DoAs from single snapshot measurements. Simulation results for evaluating the performance of the algorithm are presented in section \ref{sim}. We conclude the paper in section \ref{concl} followed by references. Proofs of some of the theorems are provided in the Appendix.
\section{Signal Model}
\label{model}
We consider an array of $M$ elements, impinged by an unknown number ($S$) of sources. The measurements at each element can be expressed as a superposition of $S$ elementary waveforms ($a(\alpha_{i},d), i = 1,2,\ldots,S$), each containing unknown angles $\alpha_i\in[\kappa_1,\kappa_2]$ as,
\begin{equation}
\tilde{b}(d) = \sum_{i=1}^{S} s_{i} a(\alpha_{i},d) + v(d),\nonumber
\end{equation}
where $v(d)$ is a white Gaussian noise process with zero mean and variance $\sigma^{2}$, $s_{i}$ are the weights and $\tilde{b}(d)$ are the measurements over the spatial variable $d=1,2,\ldots,M$. The recovery problem now reduces to detecting the number of sources $S$, estimating their corresponding weights $s_{i}$ and parameters $\alpha_{i}$, which is non-linear \cite{CBP}.

In the grid based signal model for detection and estimation, the interval $[\kappa_1,\kappa_2]$ is discretized into $N$ bins, each of size $r$ to obtain the estimation grid, ${\rho_{1},\ldots,\rho_{N}}$. Let $x_{k}$ denote the weight, corresponding to the source in $k^{th}$ bin. The discrete model approximation for $\tilde{b}(d)$ is then given by \cite{CBP,GLTLS,RJ13},
\begin{equation}
 b(d) = \sum_{k=1}^{N} x_{k} a(\rho_{k},d) + v(d).\nonumber
\end{equation}
The above equation can be expressed in vector form as:
\begin{equation}
b(d) = \mathbf{a}^{T}(d)\mathbf{x} + v(d),\nonumber
\end{equation}
where $\mathbf{x}=[x_{1},x_{2},\ldots,x_{N}]^{T}$ and $\mathbf{a}(d) =[a(\rho_{1},d), \ldots,$ $,a(\rho_{N},d)]^{T}$. Stacking the measurements, we obtain
\begin{equation}
\mathbf{b}_{M\times1} = \mathbf{A}_{M\times N}\mathbf{x}_{N\times1} + \mathbf{v}_{M\times1},
\label{CS_Noise}
\end{equation}
where $\mathbf{b}$ is the measurement vector, $\mathbf{A} = [\mathbf{a}(0), \mathbf{a}(1),\ldots,$ $, \mathbf{a}(M-1)]^{T}$ is the array steering matrix (with $M\leq N$), and $\mathbf{x}$ is the signal of interest which has a sparse or almost sparse representation under the basis of $\mathbf{A}$. 

This discretization of the estimation grid into discreet bins is the cause of grid mismatch \cite{CalBnk}. If the bin size is made too small, then there is also the risk of columns of $\mathbf{A}$ becoming correlated, thus reducing the incoherence of $\mathbf{A}$. Classical grid based estimation methods recover the desired signal without any ambiguities only if the signal is placed exactly on the corresponding grid cells. Any grid mismatch leads to ambiguities in estimation due to the leakage of source power over all the grid cells. The source power leakage depends on the kernel used for the construction of the array steering matrix, $\mathbf{A}$. In the SSR framework, it may also mean that the signal is less or even no longer sparse in the basis domain \cite{CalBnk}. Hence, it is necessary to account for grid mismatch in DoA detection and estimation.
\subsection{Modeling Grid Mismatch}
\label{gmm}
The earliest model proposed for grid matching is the errors in variables (EIV) model, which treats the grid mismatch effect as an additive error matrix, $\mathbf{E}$ as shown below in \eqref{EIV} \cite{GLTLS,TS}, 
\begin{eqnarray}
\mathbf{b} = \hat{\mathbf{A}}\mathbf{x} + \mathbf{v}, & \hat{\mathbf{A}} = \mathbf{A} + \mathbf{E}
\label{EIV}
\end{eqnarray}
The performance of the model in \eqref{EIV} characterized by its Cramer-Rao bound, derived in \cite{bcrbeiv}. However, the model described by \eqref{EIV} does not exploit the inherent Vandermonde structure of the array steering vectors in the DoA signal model, hence making the detection and estimation of DoAs complicated.

A special case of the EIV model, which preserves the Vandermonde structure of DoAs, is obtained by the Taylor series based interpolation model. This model is obtained by an interpolation of the array steering matrix, $\mathbf{A}$ w.r.t the parameters of interest as described below \cite{RJ13}.

We note that the grid mismatch problem occurs if a particular parameter of interest, $\alpha_{i}$ is not present on the estimation grid. Hence, to add $\alpha_{i}$ to the estimation grid, a Taylor series interpolation of $a(\rho,d)$ over the nearest $\rho_{k}$ is given by \cite{RJ13},
\begin{equation}
a(\rho_{k}+ p_{k},d) \approx a(\rho_{k},d) + \frac{\partial a(\rho,d)}{\partial\rho}\Bigl|_{\rho = \rho_{k}}p_{k}.\nonumber
\end{equation}
Here $p_{k}$ gives the perturbation on $\rho_{k}$ to add $\alpha_{i}$ to the grid, and it is assumed real and unknown. It can be seen that the misaligned grid can become an aligned grid if $p_{k}$ can be estimated correctly. Thus, the grid-mismatch problem can be converted into an estimation problem with interpolation over the estimation grid.

Including the Taylor series approximation, the measurements, $b(d)$ can be approximated as,
\begin{equation}
b(d) = \sum_{k=1}^{N} x_{k} a(\rho_{k},d) + \sum_{k=1}^{N} x_{k}p_{k} \frac{\partial a(\rho,d)}{\partial\rho}\Bigl|_{\rho = \rho_{k}} + v(d).\nonumber
\end{equation}
The above equation can be expressed in vector form as,
\begin{equation}
b(d) = \mathbf{a}^{T}(d)\mathbf{x} + \mathbf{a'}^{T}(d)\mathbf{P}\mathbf{x} + v(d),\nonumber
\end{equation}
where $\mathbf{a'}(d) =\Bigl[\frac{\partial a(\rho,d)}{\partial\rho}\Bigl|_{\rho = \rho_{1}},\ldots,\frac{\partial a(\rho,d)}{\partial\rho} \Bigr|_{\rho = \rho_{N}}\Bigr]^{T}$ and $\mathbf{P}= \diag(\mathbf{p})$, where $\mathbf{p}=[p_{1},\ldots,p_{N}]^{T}$ represents the grid mismatch of the estimation grid. 

Stacking the measurements, the above equation can be expressed in the matrix-vector form as,
\begin{equation}
\mathbf{b} = \mathbf{Ax} + \mathbf{A_1}\mathbf{P}\mathbf{x} + \mathbf{v},
\label{interp}
\end{equation}
where $\mathbf{A_1}=[\mathbf{a'}(0),\mathbf{a'}(1),\ldots,\mathbf{a'}(M-1)]^{T}$. So writing $\mathbf{A_1} \mathbf{P} = \mathbf{E}$, the interpolation model for grid mismatch becomes a special case of the EIV model in (\ref{EIV}). The model in \eqref{interp} has been used for deriving the Cramer-Rao bound in \cite{RJ13}, which evaluates the accuracy of the model for grid matching and hence, justifies its use.  

The model in (\ref{interp}) can be equivalently expressed as,
\begin{align}
\mathbf{b} &= \begin{bmatrix}\mathbf{A}|\mathbf{A_1}\end{bmatrix}\begin{bmatrix}\mathbf{x}\\ \mathbf{p} \odot\mathbf{x}\end{bmatrix} + \mathbf{v}, \nonumber\\
&= \boldsymbol{\Phi}\mathbf{y} + \mathbf{v};\hspace{1cm}\mathbf{y}= \begin{bmatrix}\mathbf{x}\\ \mathbf{p}\odot\mathbf{x} \end{bmatrix}.
\label{interp1}
\end{align}
Here, it should be noted that, in the above equation if $x_j=0$, for some $j\in\{1,2,\ldots,N\}$ then $p_j$ has no contribution to $\mathbf{b}$, \emph{i.e.},, by definition $p_j\neq0$ only if $x_j\neq0$. The model in \eqref{interp1} has been used for estimation of the weights $\x$ and grid matching (or estimation of $\p$), with the knowledge of the number sources ($S$) in the measurements \cite{RJ13}. In this work, we also have the additional problem of detecting the number of sources $S$ from the measurements.

We consider another equivalent model by noting that the Vandermonde structure of the array steering vectors gives $\A_{1}\P = \D\A\C$, where $\D =\diag([1,2,\ldots,M])$ and $\C=\diag(\c)$, where the entries of $\c$ depend on the geometry of the DoA problem. For the case of uniform linear array (ULA), $\c = [j2\pi\Delta\cos(\theta_1)p_1, j2\pi\Delta\cos(\theta_2)p_2, \ldots, j2\pi\Delta\cos(\theta_M)p_M]$. From the above discussions we have,
\begin{equation}
\b =\begin{bmatrix}\A|\D\A\end{bmatrix}\begin{bmatrix}\x\\\c\odot\x\end{bmatrix}+\v.
\label{interp2}
\end{equation}
We observe that any source DoA $\alpha_{i}$ can be expressed as the sum $\rho_{i}+p_{i}$, where $\rho_{i}$ lies on the estimation grid and hence can be estimated for any choice of the estimation grid. Now, we select the estimation grid ($\boldsymbol{\rho}$) of $\A$ in such a way that the array steering matrix is constrained to be orthogonal, i.e, $\A^{H}\A = \I$. This choice of the estimation grid makes $\A$ maximally incoherent and hence is the best for SSR and also has computational advantages. We now pre-multiply \eqref{interp2} by $\A^{H}$ to obtain,  
\begin{align}
\A^{H}\b &= \begin{bmatrix}\I|\A^{H}\D\A\end{bmatrix}\y + \A^{H}\v,\nonumber\\
\overline{\b} &= \begin{bmatrix}\I|\G\end{bmatrix}\overline{\y}+\overline{\v},\hspace{1cm}\overline{\y}= \begin{bmatrix}\mathbf{x}\\ \mathbf{c}\odot\mathbf{x} \end{bmatrix},
\label{interp3}
\end{align}
where, $\G = \A^{H}\D\A$, $\overline{\b} = \A^{H}\b$, $\overline{\v}=\A^{H}\v$. The model in \eqref{interp3} will be used for detection (finding $S$), estimation (estimating $\boldsymbol{\rho}$) and grid matching (estimating $\p$) of DoAs. In the case of multiple snapshot measurements, an extension of \eqref{interp3} gives a nice structure which can be used for estimation of DoAs using the techniques described in \cite{bctls}.

Let $\boldsymbol{\alpha}$ be the vector representing $S$ source locations (actual DoAs) and let $\hat{\boldsymbol{\rho}}$ represent the $\hat{S}$ location estimates of the sources. We define the probability of correct detection ($P_c$) as the probability that all the sources and their locations are detected correctly, \emph{i.e.}, $P_c = \pr\{\hat{\boldsymbol{\rho}} = \boldsymbol{\alpha}\}$, similarly the probability of miss ($P_m$) is defined as the probability that one or more sources is not detected, \emph{i.e.}, $P_m = \pr\{\hat{S}<S, \hat{\rho}_i=\alpha_i, i = 1,2, \ldots, \hat{S}\}$ and the probability of false alarm, $P_f = 1-P_c-P_m$. We define the signal to noise ratio, SNR as $\mean\{\NrmTwo{\A\x}^2\}/\mean\{\NrmTwo{\v}^2\}$. 

\emph{Problem Description:} Given the measurements $\b$, the array steering matrix $\A$, SNR and the required probability of correct detection $P_c$. The goal is to propose test statistics to detect the number of sources $\hat{S}$, their corresponding locations $\hat{\rho}_i$ on the estimation grid and estimate the grid mismatch error $p_{i}$ to match the grid. The proposed tests should achieve the required probability of correct detection $P_c$. 
\section{Joint Detection Estimation and Grid Matching for Multiple Targets}\label{main_content}
In this section, we briefly review the lasso estimator, the lasso path and propose tests for joint detection, estimation and grid matching of DoAs from single snapshot measurements. 

\emph{The Lasso Estimator}: The lasso estimator for the DoA model in \eqref{CS_Noise} is given by the solution of the following optimization problem.
\begin{equation}
\xhat(\tau) = \argmin_{\x} \frac{1}{2}\NrmTwo{\b-\A\x}^2 + \tau\NrmOne{\x},
\label{lasso}
\end{equation}
where $\xhat(\tau)$ is the estimate of $\x$ and $\tau\in[0,\infty)$ is the regularization parameter which controls the sparsity of $\xhat$. Applying KKT conditions to \eqref{lasso}, the lasso solution can be characterized as follows,
\begin{thm}
For a certain value of $\tau$, the solution to \eqref{lasso} is characterized by
\begin{align}
&\a_{i}^{H}(\b-\A\xhat) = \tau\frac{\xhats_{i}}{|\xhats_{i}|}&\forall\xhats_{i}\neq 0&,\label{sol0}\\
&|\a_{i}^{H}(\b-\A\xhat)|<\tau&\forall\xhats_{i}=0,&
\end{align}
where $\hat{x}_j$, $j = 1,2,\ldots,M$ is the $j^{th}$ entry of $\xhat$ and $\a_j$ is the $j^{th}$ column of $\A$. The singular points (knot points) occur when the second condition is changed to $\tau = \max\limits_{\{i|\xhats_{i}=0\}}|\a_{i}^{H}(\b-\A\xhat)|$.
\end{thm}
\begin{proof}
See \cite[Theorem 1]{complexlars}.
\end{proof}
We observe that the lasso solution for the special case of orthogonal array steering matrix ($\A^{H}\A = \I$) reduces to the following thresholding estimator, 
\begin{equation}
\hat{x}_j(\tau) =
\begin{cases}
\a_{j}^{H}\b-\tau\frac{\hat{x}_j}{|\hat{x}_j|} & \quad \text{if}\hspace{2mm}|\a_{j}^{H}\b|>\tau,\\
0 & \quad \text{if}\hspace{2mm}|\a_{j}^{H}\b|\leq\tau.
\end{cases}
\label{lassorth}
\end{equation}
We now discuss the behavior of $\xhat$ for variations in $\tau$, which is called the lasso path. The lasso path can be obtained using the iterative algorithm described in \cite{complexlars}.

\emph{Lasso Path}: The lasso estimator $\xhat(\tau)$ is a continuous and piecewise linear function of $\tau$. The points $\tau_k$ with $\tau_1\geq\ldots\geq\tau_k\ldots\geq\tau_r$, where the slope of the function $\xhat(\tau)$ changes are called knots (or singular points) \cite{complexlars}. For all $\tau\geq\|\A^{H}\b\|_{\infty}$, the lasso estimate $\xhat(\tau) = \mathbf{0}$. For decreasing $\tau$, each knot $\tau_k$ marks the entry or removal of some variable from the current active set ($J$), which is the index set corresponding to non-zero entries of $\hat{\x}(\tau_{k-1})$. Hence, the active set remains constant in between the knots. For a matrix $\A$ satisfying the special positive cone condition (example orthogonal matrices), no variables are removed from the active set as $\tau$ decreases and hence there are always $M$ knots in the lasso path. 

We observe that the sparsity changes only at the knots. The estimation algorithm of \cite{complexlars} sequentially iterates over the knot points, $\tau_k, k=1,2,\ldots,r$ and calculates $\xhat(\tau_k)$. So, we propose tests at the knot points to obtain a stopping condition for the iterative algorithm as the lasso solution varies from $\xhat(\tau_1)$ to $\xhat(\tau_S)$. Once, the tests detect the number of sources $\hat{S}$ or equivalently $\tau_{\hat{S}}$, the DoAs can then be estimated by solving lasso at $\tau =\tau_{\hat{S}}$. 
\subsection{Orthogonal Models}
\label{orthogonal}
Here we assume that the array steering matrix is orthogonal ($\A^{H}\A = \I$) and the sources lie on the estimation grid (perfect grid matching). These assumptions make the analysis of the test statistics simpler for evaluating thresholds. Specifically, the components of the lasso estimate, $\xhat$ in \eqref{lassorth} are independent. Although, this scenario is not practical as it occurs only for antennas with infinite apertures, the insights obtained here are helpful in proposing tests while working with non-orthogonal (over-sampled) and grid matching models. In the following, we propose the covariance test, test-A, test-B and test-C. The first three tests also require the additional knowledge of noise variance.

\emph{Covariance Test:} The covariance test statistics is defined at the knots of the lasso path. At the $k^{th}$ knot, the covariance test statistics is defined as \cite{siglass},
\begin{equation}
T_k = \frac{1}{\sigma^2}\Big(\b^{H}\A\xhat(\tau_{k+1})-\b^{H}\A_{J}\tilde{\x}_{J}(\tau_{k+1})\Big),
\label{covtest}
\end{equation}
where $J$ is the active set just before $\tau_k$, $\tilde{\x}(\tau_{k+1})$ is the solution of the lasso problem using only the active model $\A_{J}$ (columns of $\A$ belonging to $J$), with $\tau = \tau_{k+1}$, \emph{i.e.},
\begin{equation}
\tilde{\x}_{J}(\tau_{k+1}) = \argmin_{\x\in\Re^{|J|}}\frac{1}{2}\NrmTwo{\b-\A_{J}\x_{J}}^2 + \tau_{k+1}\NrmOne{\x_{J}}.
\end{equation}
Intuitively, the covariance test statistics defined in (\ref{covtest}) is a function of the difference between $\A\xhat$ and $\A_{J}\tilde{\x}_{J}$, which represents the fitted values of the model by including and leaving out the next $\hat{x}_j$ (corresponding to the knot at $\tau_{k+1}$), respectively. For the case of orthogonal $\A$, it can be shown \cite[Lemma 1]{siglass} that the covariance test statistics reduces to
\begin{equation}
T_k = \tau_k(\tau_{k}-\tau_{k+1})/\sigma^2,\hspace{1mm}k = 1,2,\ldots,M-1,
\end{equation}
where, the $M$ knots of the lasso estimator $\xhat(\tau)$ are given by $[\mathcal{I}, \boldsymbol{\tau}] =$sort($|\A^{H}\b|$). The function sort($\u$) sorts the entries of $\u$ in the descending order, $\mathcal{I}$ is the collection of the corresponding indices of $|\A^{H}\b|$ and $\boldsymbol{\tau}$ is the vector of $M$ knot points.

Now, let the number of non zero entries in the actual parameter $\x$ be $S$. We define $B$ as the event that the $S$ sources are added to the estimate $\xhat$ at the first $S$ knot points of the lasso path:
\begin{equation}
B = \Big\{\min_{j\in \tilde{T}}\tau_{j}>\max_{j\notin \tilde{T}}\tau_{j}\Big\}.
\end{equation}
For the case of orthogonal models, event $B$ reduces to
\begin{equation}
B = \Big\{\min_{j\in \tilde{T}}|\a_{j}^{H}\b|>\max_{j\notin \tilde{T}}|\a_{j}^{H}\b|\Big\},
\end{equation}
where $\tilde{T}$ is the support of the original parameter $\x$ (columns of $\A$ corresponding to non-zero entries of $\x$).

\emph{Remark-1}: Event $B$ is defined to ensure that $S$ active parameters ($S$ sources) are added to the estimate $\xhat$ in the first $S$ knots, then the test statistics at $(S+1)^{th}$ knot and beyond would depend only on the truly inactive variables (noise). The detection tests proposed below are conditioned on event $B$. Hence, $P(B)=1$ is a sufficient condition for the detection tests to provide rate control ($\hat{P}_c = P_c$). However, we show in Lemma \ref{lem1} that $P(B)\to1$, whenever the power of the weakest source is large compared to the noise power or whenever the detection is performed in the moderate to high SNR regime \cite[Theorem-1]{siglass}. Hence, detection at moderate to high SNR is a sufficient condition for $P(B)\to1$ and hence is also a sufficient condition for the tests to provide rate control for a given $P_c$.
\begin{lem}\label{lem1}
For orthogonal models, $P(B)\to1$ at moderate to high SNRs.
\end{lem}
\begin{proof}
See Appendix-\ref{A0}
\end{proof}
From the above discussions, we conclude that it suffices to stop at the $(S+1)^{th}$ knot for providing rate control at moderate to high SNR regime. This requires the evaluation of c.d.f of $T_{\mathsmaller{S+1}}$ conditional on event $B$, given by
\begin{thm}
\label{Finitesamplecov}
The c.d.f of $T_{\mathsmaller{S+1}}$, conditional on event $B$ is,
\begin{equation*}
F_{T_{\mathsmaller{S+1}}}(\eta) = 1-n\INT{\sqrt{\eta}}{\infty}{ye^{(-y^2/2)}\left(1-e^{\frac{-(y-\eta/y)^2}{2}}\right)^{n-1}\mathrm{d}y},
\end{equation*}
where $n = M-S$.
\end{thm}
\begin{proof}
See Appendix-\ref{A1}
\end{proof}
Now, with the knowledge of the c.d.f of $T_{S+1}$ conditional on event $B$, the problem of finding the number of sources $S$ reduces to the following hypothesis testing problem.
\begin{center}
$H_o = T_k$ is distributed as $F_{T_{\mathsmaller{S+1}}}$.\\
$H_a = T_k$  is not distributed as $F_{T_{\mathsmaller{S+1}}}$.\\
\end{center}
The idea is to evaluate the test statistics at each knot in the increasing order (from $\tau_M$ to $\tau_1$) and compare the value to a threshold, $\eta$. The first instance, where $T_k>\eta$ is the stopping point, because conditional on $B$, the stopping point corresponds to the knot $\tau_S$, where all the sources have been added to the lasso estimate $\xhat$. The threshold, $\eta$ is obtained from the tail probability of the c.d.f of $T_{S+1}$  by fixing the required probability of correct detection, $P_c$
\begin{equation}
P_c = \pr\{T_{k}\leq\eta\} = F_{T_{S+1}}(\eta).
\end{equation}
We observe that the c.d.f of the covariance test, though an exact (non-asymptotic) distribution, requires numerical integration for evaluating the threshold at each knot, hence making the test complicated. In \cite{siglass}, the asymptotic c.d.f of $T_{k}, k>S$, conditional on event $B$ is derived for real measurement model. The extension to complex measurement model is given by the following theorem,
\begin{thm}\label{asympcov}
Let the magnitude of the smallest nonzero entry of $\x$ be large compared to $\sigma$. Then event $B$ is satisfied, \emph{i.e.}, $\pr(B)\rightarrow 1$ and furthermore, for each fixed $l\geq0$
\begin{equation*}
[T_{\!\mathsmaller{S+1}},T_{\!\mathsmaller{S+2}},\ldots,
T_{\!\mathsmaller{S+l}}]\xrightarrow{d}\left[\mathrm{Exp}(1),\mathrm{Exp}(\frac{1}{2}),\ldots,
\mathrm{Exp}(\frac{1}{l})\right],
\end{equation*}
conditional on $B$, \emph{i.e.}, the $l^{th}$ statistics after $S$ converges independently to exponential distribution with mean $1/l$.
\end{thm}
\begin{proof}
See Appendix-\ref{A2}
\end{proof}
We observe that although the asymptotic distribution of $T_{S+1}$ is tractable, it converges very slowly ($2\log M$), hence offering lesser control in-terms of $P_c$. So we now propose other tests which are both easy to evaluate and exact.

\emph{Test-A:} We note that, if event $B$ is satisfied and there are $S$ sources, then $A_k=\frac{\tau_{S+k}}{\sigma}, k=1\ldots,M-S$ are the order statistics of Rayleigh random variables. We define the Rayleigh test statistics as
\begin{equation}
A_{k} = \frac{\tau_{k+S}}{\sigma}.
\label{raytest}
\end{equation}
We note that $\tau_{S+1}$ is the first knot point corresponding to noise, conditional on event $B$. Hence, $P_c$ can be controlled by accurately detecting $A_1$. The threshold for controlling $P_c$ requires the c.d.f of $A_{1}$ which is given by,
\begin{thm}
The c.d.f of $A_1$ conditional on event $B$ is, 
\begin{equation}
F_{A_{1}}(x)= (1-\exp(-x^2/2))^{M-S}.
\label{orderstat}
\end{equation}
\end{thm}
\begin{proof}
$A_{1}$ is the maximum of the i.i.d Rayleigh random variables and hence its c.d.f is obtained by \eqref{orderstat}.
\end{proof}
The problem of finding $S$ sources reduces to comparing $A_{k}$ with a threshold ($\eta$) at each knot point. The threshold is obtained from the c.d.f \eqref{orderstat} by fixing $F_{A_{1}}$ to the required $P_c$. 

\emph{Test-B:} Let us consider the random variables $E_{i} = A^{2}_{i}, i = 1, 2, \ldots, n$. Then $E_i$ are the order statistics of the standard exponential distribution, conditional on event $B$. Now, we define the Exponential test statistics $B_{n}=E_{n}-E_{n-1}$. The c.d.f of $B_{S+1}$, conditional on event $B$ is required for detection of $S$ sources, which is given by,
\begin{thm}
The c.d.f of $B_{\mathsmaller{S+1}}$ conditional on event $B$ is,
\begin{equation}
F_{B_{\mathsmaller{S+1}}}(x) = 1-\exp(-x).
\label{gammapdf}
\end{equation}
\end{thm}
\begin{proof}
$E_i$ are the order statistics of the standard exponential distribution. The c.d.f of $B_{\mathsmaller{S+1}}$ can now be obtained as follows. The joint pdf of $B_{\mathsmaller{n}}$ and $E_n$ is,
\begin{equation*}
f_{B_{\mathsmaller{n}},E_n}(g,y) = C\{F(y-g)\}^{n-2}f(y-g)f(y).
\end{equation*}
Hence the pdf of the test statistics $B_{n}$ is
\begin{equation*}
f_{B_{\mathsmaller{n}}}(g) = \int_{0}^{\infty}C\{F(y-g)\}^{n-2}f(y-g)f(y)dy.
\end{equation*}
The cdf of the test statistics $G$ is given by
\begin{align*}
F_{B_{\mathsmaller{n}}}(\eta) &= \int_{0}^{\eta}\int_{0}^{\infty}C\{F(y-g)\}^{n-2}f(y-g)f(y)dydg\\
&=1-\exp(-\eta), n = S+1,\ldots,M.
\end{align*}
\end{proof}
Again, the problem of finding $S$ sources reduces to comparing $B_{k}$ with a threshold ($\eta$) at each knot point. The threshold is obtained from the c.d.f \eqref{gammapdf} by fixing $F_{B_{\mathsmaller{S+1}}}$ to required $P_c$.  
\subsubsection{Unknown noise variance}
Here, we propose a test statistics for the case when the noise variance is unknown and needs to be estimated. We retain the orthogonality and perfect grid matching assumptions discussed at the beginning.

\emph{Test-C:} We choose the estimate of the noise variance as $\hat{\sigma}^2 = \|\b-\A\xhat_{I}\|^{2}_{2}$, where $\xhat_{I}$ is the least-square estimate using the model after $(M-1)$ steps of Algorithm-\ref{Algo1}. The reason for the choice of using $(M-1)$ supports for estimating variance is that it is well known that an antenna array of $M$ elements can recover at-most $(M-1)$ sources \cite{van2002optimum}, hence the effect of all the sources impinging the array is removed from the measurements after $M-1$ steps. Now, we propose the test statistics at the $k^{th}$ knot as,
\begin{flalign}
&C_{k} =\frac{\tau^{2}_{k+S}}{\hat{\sigma}^{2}} =\frac{A^{2}_{k}}{\hat{\sigma}^2/\sigma^2}, k=1,\ldots,l-1, &
\end{flalign}
where $l = M-S$. The distribution of $C_1$ conditional on event $B$ is required for detecting the $S$ sources and is given by,
\begin{thm}
The c.d.f of $C_{1}$, conditional on event $B$ is,
\begin{equation}
F_{C_{1}}(\eta) = \SUM{r}{0}{l}{(-1)^{r}\binom{l}{r}\left(1+r\eta\right)^{-1}}.
\label{Ftest}
\end{equation}
\end{thm}
\begin{proof}
We observe that $\frac{\hat{\sigma}^2}{2\sigma^2}$ is a $\chi^2$ random variable with $2$ degrees of freedom for all $k$, \emph{i.e.}, $\hat{\sigma}^{2}\backsim\chi_{2}^{2}$ and $R_{k}^{2}/2,k=1,2,\ldots,l-1$ are the order statistics of a $\chi^2$ random variable with $2$ degree of freedom. Hence, $C_{1}$ is the maximum of $F$ random variables with equal correlations, whose distribution is given by \cite{Gupte,osyoung},
\begin{equation*}
F_{C_{1}}(\eta) = \SUM{r}{0}{k}{(-1)^{r}\binom{k}{r}\left(1+r\eta\right)^{-1}}
\end{equation*}
\end{proof}
Again, the problem of finding $S$ sources reduces to comparing $C_{k}$ with a threshold ($\eta$) at each knot point. The threshold is obtained from the c.d.f \eqref{Ftest} by fixing $F_{C_{1}}$ to required $P_c$. We note that test-$C$ proposed here is very similar to the to $F$ test used in the least squares regression for selecting the best model. However, the main difference is that the threshold in the least squares regression setup is evaluated by observing the degree of the $F$ random variable at each step, whereas here we show that the maximum of equicorrelated $F$ random variables is a better test statistics for evaluating the threshold. We summarize the steps for detection and estimation of DoAs with orthogonal measurement model in Algorithm-\ref{Algo1} using test-$A$ as an example. All the other tests described earlier can be implemented by evaluating the corresponding test statistics in step-$3$ of the algorithm.
\begin{algorithm}
\caption{Algorithm for Detection and Estimation}
\label{Algo1}
\begin{algorithmic}[1]
\State\textbf{Inputs:} $\b$, $\A$, $\boldsymbol{\eta}$ (obtained by inverting the c.d.f).
\State\textbf{Initialize:} Set $i=M-1$, $\hat{S} = 0$, [$\mathcal{I}$,$\boldsymbol{\tau}] =$ sort($|\A^{H}\b|$).
\State\textbf{Evaluate:} Evaluate the test statistics $A_i$.
\State\textbf{Decision: If} $A_i\geq\eta_i$ \textbf{go to} step $6$
\State\textbf {Iterate:} Decrease $i$ by 1 and iterate from step $3$.
\State\textbf{Outputs:} $\hat{S}= i$, $\mathsmaller{\hat{T}=\mathcal{I}(1,2,\ldots,\hat{S})}$, $\hat{\tau} = \boldsymbol{\tau}(\hat{S})$, $\hat{\boldsymbol{\rho}}=\boldsymbol{\rho}(\hat{T})$.
\end{algorithmic}
\end{algorithm}

\subsubsection{Low SNR scenarios}
We observe that the tests proposed for orthogonal models require the probability of event $B$ to be close to $1$ (\emph{i.e.}, $P(B)\rightarrow 1$) for obtaining rate control w.r.t $P_c$. For orthogonal models, it was shown that moderate to high SNR scenarios are sufficient for $P(B)\rightarrow 1$. Here, we make some comments on low SNR scenarios and explain the difficulty for proposing tests at low SNR scenarios. 

We observe that the tests discussed above depended on some functions of the p.d.f of the estimator, $\xhat$. For e.g., the knot points correspond to singularities of $\xhat$. So it would be useful to consider the p.d.f. of the lasso estimator. For a real linear model in real Gaussian noise we have,
\begin{thm}
The p.d.f of the lasso estimator, $\xhat$ for orthogonal models ($\A^{H}\A=\I$) is given by, 
\begin{flalign}
&f_{\xhats_k}(\xhats_k) =  \left\{
\begin{array}{rl}
\frac{1}{\sqrt{2\pi\sigma^2}}\exp\big(-\frac{(\xhats_k+\tau-x_k)^2}{2\sigma^2}\big)& \text{if } \xhats_k>0,\\
\frac{1}{\sqrt{2\pi\sigma^2}}\exp\big(-\frac{(\xhats_k-\tau-x_k)^2}{2\sigma^2}\big) & \text{if } \xhats_k<0.
\end{array}\right. &
\label{lassopdf_orth}
\end{flalign}
\end{thm}
\begin{proof}
See \cite{RJ16lassodistarx}.
\end{proof}
We observe from \eqref{lassopdf_orth} that the p.d.f of the lasso estimator is continuous function of $\xhats_{k}$ except for the discontinuities at the knot points (when $\xhats_{k} = 0$ or $\tau = \tau_{k} = |\a_{k}^{H}\b|$). In order to understand the problems for proposing tests at low SNR, we study the expression for probability of error $P_e$ given by,
\begin{flalign*}
&P_e = P_{m} + P_{f}&\\
&= \pr(\xhats_{k}=0|x_{k}\neq 0) + \pr(\xhats_{k}\neq 0|x_{k}=0)&\\
&=\Phi(\frac{\tau-x_{k}}{2\sigma})-\Phi(\frac{-\tau-x_{k}}{2\sigma})+ G(\tau),& 
\end{flalign*}
where, $\Phi(.)$ denotes the cdf of normal random variable and $G(.)$ is a function of $\tau$ only. In the above expression, we observe that $P_m$ is a function of both $x_k$ (unknown) and $\tau$, whereas $P_f$ is only a function of $\tau$. This dependence of $P_e$ (obtained from p.d.f of $\xhats_k$) on the unknown parameter $x_k$ makes it difficult for proposing test statistics to control $P_c$. Hence, conditioning tests over event $B$ translates to assuming that $P_m\rightarrow 0$ as $\sigma\rightarrow 0$ (or moderate to high SNR), which is a good assumption for orthogonal models. We also observe that there is still complete control over $P_f$ for orthogonal models, which is usually the main objective in classical hypothesis testing. Finally, we note that controlling $P_m$ requires the prior knowledge of $\x_{k}$, which is possible in communication scenario wherein $\x$ are symbols transmitted from a predefined code-book. Hence, in a communication scenario, it may be possible to calculate exact expressions for $P_{m}$ (and $P_c$).
\subsection{Non-Orthogonal Models}
\label{conventionalmodel}
We now obtain tests for the case where the estimation grid is over-sampled to $N>>M$ bins to obtain a fat array steering matrix ($\A$). We retain the assumption that all the source locations are perfectly matched to the estimation grid. From the discussions on orthogonal models, we observed that test statistics to control $P_c$ can be proposed at knot points. Hence, we will first study the knot points of the lasso for a fat matrix $\A$. The first knot point of the lasso occurs at $\tau_{1} = \max\limits_{k} |\a_{k}^{H}\b|$. The process of finding the subsequent knots is summarized in Algorithm-\ref{Algo2}.

\emph{Remark-2}: For non-orthogonal model, we observe from simulations (section-\ref{sim}) that the following two sufficient conditions are required for $\pr(B)\to1$. Firstly, the power of the weakest source should be large compared to the noise power or the detection should be performed in the moderate to high SNR regime. Secondly, the sources should be well separated.  

We now propose a test at the knot points. The goal of the proposed test is to detect the $(S+1)^{th}$ knot point (where $S$ is unknown), conditional on event $B$.

\emph{Test-D:} The $D$ test statistics at the $k^{th}$ knot is defined as,
\begin{equation}
D_{k} = \frac{\tau^{2}_{k}}{\sigma^2}.
\end{equation}
Again, assuming event $B$ is true (\emph{i.e.}, $P(B)\to1$), we need to make a decision at $(S+1)^{th}$ knot. Hence, we require the c.d.f of $D_{1}$, given by
\begin{thm}\label{oversamp}
The c.d.f of $D_{1}$, conditional on event $B$ is,
\begin{equation}
F_{D_{1}}(\eta) = \PROD{i}{1}{M-S}{(1-e^{-\eta/\varrho_{i}})},
\label{testd}
\end{equation}
where $\varrho_{i}$ are the $M-S$ non-zero eigen values of the matrix $\Q_{\mathsmaller{M-S}}$, whose construction is described in the proof.
\end{thm}
\begin{proof}
See Appendix-\ref{A3}
\end{proof}
Similar to other tests, the problem of finding $S$ sources reduces to comparing $D_{k}$ with a threshold ($\eta$) at each knot point. The threshold is obtained from the c.d.f \eqref{testd} by fixing $F_{D_{1}}$ to the required $P_c$.
\begin{algorithm}
\caption{Algorithm for Detection and Estimation}
\label{Algo2}
\begin{algorithmic}[1]
\State\textbf{Inputs:} $\b$, $\A$, $\boldsymbol{\eta}$ (obtained by inverting the c.d.f).
\State\textbf{Initialize:} Set $k=1$, $\hat{S} = 0$, $\tau_{1} =  \max\limits_{k} |\a_{k}^{H}\b|$.
\State The active set $J = \{j_{1},j_{2},\ldots,j_{n}\}$ is determined by solving \eqref{sol0} at $\tau_{k}$.
\State For each $k\notin J$, solve the following system of equations for a vector $\xhat =[\xhats_{1},\ldots,\xhats_{n}]$ and a set $\Lambda_{k}$. 
\begin{equation*}
\left\lbrace \a_{j_{l}}^{H}(\b-\A_{J}\xhat)=\Lambda_{k}\frac{\xhats_{l}}{|\xhats_{l}|}\right\rbrace_{l=1}^{n},|\a_{j}^{H}(\b-\A_{J}\xhat)|=\Lambda_{k}
\end{equation*}
If the system is infeasible, we put $\Lambda_{j}=0$.
\State\textbf{Evaluate:} Evaluate the test statistics $D_k$.
\State\textbf{Decision: If} $D_k\geq\eta_i$ \textbf{go to} step $8$
\State\textbf{Iterate:} The next knot is given by, $\tau_{k+1} = \max\limits_{k}\Lambda_{j}$.
\State\textbf{Outputs:} $\hat{S}= i$, $\mathsmaller{\hat{T}=\mathcal{I}(1,2,\ldots,\hat{S})}$, $\hat{\tau} = \boldsymbol{\tau}(\hat{S})$, $\hat{\boldsymbol{\rho}}=\boldsymbol{\rho}(\hat{T})$.
\end{algorithmic}
\end{algorithm}

\subsection{Grid Matching}
\label{gmdetect}
For accurate detection and estimation of sources, we require the source locations to be matched with the estimation grid. The popular way to deal with the grid-mismatch problem in practice is to over-sample the estimation grid into $N >> M$ bins to obtain a fat array steering matrix, $\A$ and hope that all the source locations are perfectly matched to the estimation grid of $\A$. However, as discussed in Section-\ref{model}, it has been shown in \cite{CalBnk} that fine sampling of the estimation grid does not necessarily guarantee perfect grid matching. There is always a non-zro probability that all the sources are not aligned on the estimation grid. Moreover, there is also the problem of columns of $\A$ becoming correlated, thus reducing its incoherence. This also means that the signal is less sparse or even no longer sparse in the spatial domain \cite{CalBnk,TS}. Hence, the grid matching model discussed in Section-\ref{model} may be used for detection and estimation of off-grid sources.

The block sparse estimator for parameter estimation and grid matching can be formulated as the following group lasso optimization problem,
\begin{equation}
\hat{\y} = \arg\min_{\y}\hspace{1mm}\frac{1}{2}\NrmTwo{\overline{\b}-\P\y}^2+\tau\SUM{g}{1}{N}{\NrmTwo{\y_{g}}},\label{SOCP}
\end{equation}
where $\mathbf{y}_g = [x_g,p_gx_g]^{T}$, $g = 1,2,\ldots,N$ and $\P = [\I|\G]$. We now obtain the optimality conditions for the above optimization as,  
\begin{thm}
The solution of the group-lasso estimator satisfies the following K.K.T conditions
\begin{align}
&\P_{g}^{H}(\overline{\b}-\P\yhat) = \tau\frac{\yhat_{g}}{\|\yhat_{g}\|_{2}}&\forall\yhat_{g}\neq \mathbf{0},&\label{glass}\\
&\|\P_{g}^{H}(\overline{\b}-\P\yhat)\|_{2}\leq\tau&\forall\yhat_{g}=\mathbf{0}.&
\end{align}
where $\P_{g} = [\e_{g}|\g_{g}]$, $\e_g$ and $\g_g$ are $g^{th}$ column of $\I$ and $\G$.
\end{thm}
\begin{proof}
See \cite{glassoorigin}
\end{proof}
We can immediately notice that the first knot point is given by $\max\limits_{g}\|\P_{g}^{H}\overline{\b}\|_{2}$. The process of finding the knots is summarized in Algorithm-\ref{Algo3}. 

\emph{Remark-3}: For grid matching model, we observe from simulations (section-\ref{sim}) that $P(B)\to1$ if the sufficient conditions mentioned in \emph{Remark-2} are satisfied.
\begin{algorithm}
\caption{Algorithm for Detection and Estimation}
\label{Algo3}
\begin{algorithmic}[1]
\State\textbf{Inputs:} $\b$, $\A$, $\boldsymbol{\eta}$ (obtained by inverting the c.d.f).
\State\textbf{Initialize:} Set $k=1$, $\hat{S} = 0$, $\tau_{1} =  \max\limits_{k} \|\P_{j_{k}}^{H}\overline{\b}\|_{2}$.
\State The active groups $J = \{j_{1},j_{2},\ldots,j_{n}\}$ is determined by solving \eqref{glass} at $\tau_{k}$.
\State  For each $k\notin J$ solve the following system of equations for the blocks $\yhat =[\yhat_{1},\ldots,\yhat_{n}]$ and a set $\Lambda_{k}$. 
\begin{equation*}
\left\lbrace \P_{j_{l}}^{H}(\overline{\b}-\P_{J}\yhat)=\Lambda_{k}\frac{\yhat_{l}}{\|\yhat_{l}\|}\right\rbrace_{l=1}^{n},\|\P_{j}^{H}(\overline{\b}-\P_{J}\yhat)\|=\Lambda_{k}
\end{equation*}
If the system is in-feasible, we put $\Lambda_{j}=0$.
\State\textbf{Evaluate:} Evaluate the test statistics $E_k$.
\State\textbf{Decision: If} $E_k\geq\eta_i$ \textbf{go to} step $8$
\State\textbf{Iterate:} The next knot is given by, $\tau_{k+1} = \max\limits_{k}\Lambda_{k}$.
\State\textbf{Outputs:} $\hat{S}= i$, $\mathsmaller{\hat{T}=\mathcal{I}(1,2,\ldots,\hat{S})}$, $\hat{\tau} = \boldsymbol{\tau}(\hat{S})$, $\hat{\boldsymbol{\rho}}=\boldsymbol{\rho}(\hat{T})$.
\end{algorithmic}
\end{algorithm}


We now propose a test at the knot points of the group-lasso path. The goal of the proposed test is to detect the $(S+1)^{th}$ knot point (where $S$ is unknown), conditional on event $B$.

\emph{Test-E:} The $E$ test statistics at the $k^{th}$ knot is defined as,
\begin{equation}
E_{k} = \frac{\tau_{S+1}}{\sigma^2}, k=1,2,\ldots,M-S-1,
\end{equation}
Again, assuming event $B$ is true (\emph{i.e.},, $P(B)\to1$), we need to make a decision at $(S+1)^{th}$ knot. Hence, we require the c.d.f of $E_{\mathsmaller{S+1}}$, given by
\begin{thm}\label{testGMM}
The c.d.f of $E_{1}$, conditional on event $B$ is,
\begin{equation}
F_{E_{1}}(\eta) = \PROD{i}{1}{M-S}{\frac{\varrho_{i}}{\varrho_{i}-\varepsilon_{i}}(1-e^{-\eta/\varrho_{i}})-\frac{\varepsilon_{i}}{\varrho_{i}-\varepsilon_{i}}(1-e^{-\eta/\varepsilon_{i}})},
\label{testE}
\end{equation}
where $\varrho_{i}\geq\varepsilon_{i}, i = 1,2,\ldots,M-S$ are the $M-S$ non-zero eigen values corresponding to the $\chi^{2}$ random variables as described in the proof.
\end{thm}
\begin{proof}
See Appendix-\ref{A4}.
\end{proof}
Similar to other tests, the problem of finding $S$ sources reduces to comparing $E_{k}$ with a threshold ($\eta$) at each knot point. The threshold is obtained from the c.d.f \eqref{testE} by fixing $F_{E_{1}}$ to the required $P_c$. 
\section{Numerical Simulations}
\label{sim}
In this section, we evaluate the performance of the proposed joint detection, estimation and grid matching algorithms discussed in the previous section. In the following, we will first discuss the simulation set-up, present the results obtained by the algorithms and interpret the results. 

\subsection{Simulation Setup}
The simulation setup consists of a uniform linear array (ULA) with $M = 8$ antennas, which is receiving signal from $S$ sources \cite{van2002detection}. The sources are chosen such that the total source power, $\mean\{\NrmTwo{\x}^2\} = 1$. In the case of multiple sources, all the sources are assumed to have equal power. We generate the estimation grid $\boldsymbol{\rho}$ by uniformly sampling the interval $[-\pi/2,\pi/2]$ into $N=8$ bins for orthogonal and grid matching models and $N=16$ bins for the non-orthogonal model. The array steering matrix, $\A$ of size $M\times N$ is then generated as explained in section \ref{model}. $\A$ is further normalized to avoid gain at the receiver. The Gaussian noise is generated by selecting the noise variance based on the given value of SNR (defined in section \ref{model}). 

The sources are detected and estimated as described in Algorithm-\ref{Algo1} to Algorithm-\ref{Algo3}. Grid matching is also performed while simulating the grid matching scenario using Algorithm-\ref{Algo3}. The threshold for all the simulations is set to maintain the required probability of correct detection of $P_c = 0.99$. In the following, we use Monte-Carlo simulations for $L = 10^5$ noisy realizations to evaluate the performance. We also calculate $\pr(B)$ for different scenarios by checking if the knot points corresponding to the sources occur first in the lasso path.
\subsection{Orthogonal Models}
In the simulations for the single source scenario, the source impinges the ULA from the angle $\boldsymbol{\rho}(5)$. Similarly, sources impinge ULA from angles $\boldsymbol{\rho}(3,6)$ for two source scenario. For three and four source scenarios the sources impinge the ULA from $\boldsymbol{\rho}(2,4,6)$ and $\boldsymbol{\rho}(2,4,6,7)$ angles respectively. Figure-\ref{fig:Fig01} shows the plot of $\pr(B)$ vs SNR. Tables (\ref{Ta1}-\ref{Ta5}) show the $\hat{P}_c$ obtained by Algorithm-\ref{Algo1} based on tests mentioned in the caption of the table (Asymptotic covariance test, Exact covariance test, Test-A, Test-B, Test-C). For Test-C, the noise variance is unknown and is estimated as described earlier. The number of sources ($S$) received are indicated in the sub-caption. The following observations can be made from Fig. \ref{fig:Fig01} and Tables-(\ref{Ta1}-\ref{Ta5}).
\begin{enumerate}
\item Fig. \ref{fig:Fig01} shows that $\pr(B)\to1$ for SNR $>10,15,20,20$ dB for one, two, three and four source scenarios. 
\item None of the proposed tests provide rate control (\emph{i.e.}, $\hat{P}_c<P_{c}$) for SNR$<15$ dB for single source, SNR$<20$ dB for two source and three source, and SNR$<25$dB for four source scenarios respectively. The reason for this behaviour is that $P(B)\neq1$ in these scenarios, so the tests fail when SNR is low.
\item All the finite sample tests ($A_k$, $B_k$ and $T_{k}$ (Finite)) give perfect rate control ($\hat{P}_c=P_{c}$) independent of SNR whenever SNR$\geq15$, SNR$\geq20$, SNR$\geq20$ dB and SNR$\geq25$ dB for single, two, three and four source scenarios respectively.
\item The asymptotic covariance test ($T_k$ (Asymp)) does not give rate control (\emph{i.e.}, $\hat{P}_c<P_{c}$) even at high SNRs.
\item Test-C ($C_k$) also provides rate control for high SNRs for all the scenarios.
\item $P_f\leq 0.01$ for exact covariance test, Test-A and Test-B at moderate to high SNRs.
\item The Algorithm also outputs the correct DoA(s) for $99\%$ ($P_c$) of the trials at moderate to high SNRs. 
\end{enumerate}
From the observations, we can conclude that $\pr(B)\to1$ at moderate to high SNRs for orthogonal models, thus verifying Lemma-\ref{lem1}. We also conclude that the proposed finite sample tests maintain rate control ($\hat{P}_c= P_c$) at moderate to high SNRs, where $\pr(B)\to1$. Whenever the source is detected correctly, the estimation error is zero because of perfect grid matching. We note that the evaluation of threshold ($\eta$) for the finite sample covariance test requires numerical integration, which makes it the most complex test. But, there is no gain in-terms of rate control compared to the other finite sample tests.  We also note that although the tests have been performed for $P_c = 0.99$, the rate control for higher values of $P_c$ was also observed and upto $7$ sources could be detected. 
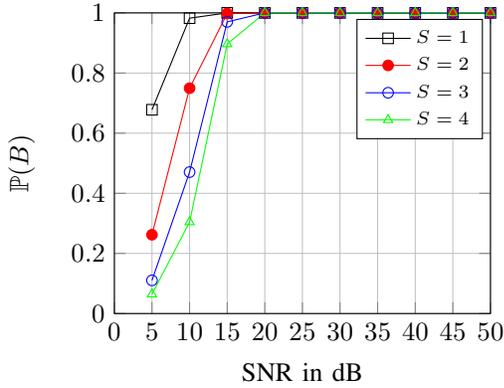
\begin{figure}[h]
\centering
\begin{tikzpicture}
\draw[help lines] (0,0);
\begin{axis}[%
width=5cm,
height=4cm,
scale only axis,
xmin=0,
xmax=50,
xtick distance={5},
grid=both,
grid style={line width=.1pt, draw=gray!10},
major grid style={line width=.2pt,draw=gray!50},
xlabel={SNR in dB},
ymin=0,
ymax=1,
ylabel={$\pr(B)$},
]


\addplot [
color=black,
solid,
mark=square,
mark options={solid}
]
table[row sep=crcr]{
   5.0000    0.6778\\
   10.0000   0.9824\\
   15.0000   1.0000\\
   20.0000   1.0000\\
   25.0000   1.0000\\
   30.0000   1.0000\\
   35.0000   1.0000\\
   40.0000   1.0000\\
   45.0000   1.0000\\
   50.0000   1.0000\\
};
\addlegendentry{\scriptsize{$S=1$}};
\addplot [
color=red,
solid,
mark=*,
mark options={solid}
]
table[row sep=crcr]{
   5.0000     0.2621\\
   10.0000    0.7496\\
   15.0000    0.9979\\
   20.0000    1.0000\\
   25.0000    1.0000\\
   30.0000    1.0000\\
   35.0000    1.0000\\
   40.0000    1.0000\\
   45.0000    1.0000\\
   50.0000    1.0000\\
};
\addlegendentry{\scriptsize{$S=2$}};
\addplot [
color=blue,
solid,
mark=o,
mark options={solid}
]
table[row sep=crcr]{
   5.0000     0.1104\\
   10.0000    0.4705\\
   15.0000    0.9701\\
   20.0000    1.0000\\
   25.0000    1.0000\\
   30.0000    1.0000\\
   35.0000    1.0000\\
   40.0000    1.0000\\
   45.0000    1.0000\\
   50.0000    1.0000\\
};
\addlegendentry{\scriptsize{$S=3$}};
\addplot [
color=green,
solid,
mark=triangle,
mark options={solid}
]
table[row sep=crcr]{
   5.0000     0.0644\\
   10.0000    0.3042\\
   15.0000    0.8964\\
   20.0000    1.0000\\
   25.0000    1.0000\\
   30.0000    1.0000\\
   35.0000    1.0000\\
   40.0000    1.0000\\
   45.0000    1.0000\\
   50.0000    1.0000\\
};
\addlegendentry{\scriptsize{$S=4$}};
\end{axis}
\end{tikzpicture}%
\caption{$\pr(B)$ vs SNR for orthogonal model}
\label{fig:Fig01}
\end{figure}
\subsection{Non-Orthogonal Models}
In the simulations for the single source scenario, the source impinges the ULA from the angle $\boldsymbol{\rho}(9)$. Similarly, sources impinge ULA from angles $\boldsymbol{\rho}(7,10)$ for two source scenario. For three and four source scenarios the sources impinge the ULA from $\boldsymbol{\rho}(6,9,12)$ and $\boldsymbol{\rho}(5,8,11,14)$ angles respectively. Fig. \ref{fig:Fig02} shows the plot of $\pr(B)$ vs SNR.  Table-\ref{Ta6} shows the $\hat{P}_c$ obtained by Algorithm-\ref{Algo2} based on Test-D. The number of sources ($S$) received are indicated in the sub-caption. The following observations are made from Fig. \ref{fig:Fig02} and Table-\ref{Ta6},
\begin{enumerate}
\item Fig. \ref{fig:Fig02} shows that $\pr(B)\to1$ for SNR $>15,20,25,35$ dB for one, two, three and four source scenarios. 
\item Test-D does not provide rate control (\emph{i.e.}, $\hat{P}_c<P_{c}$) for SNR$<15$ dB for single source, SNR$<22$ dB for two source, SNR$<25$ dB for three source, and SNR$<35$dB for four source scenarios respectively. The reason for this behaviour is that $P(B)\neq1$ in these scenarios, so the tests fail when SNR is low and sources are not well separated.
\item The test gives perfect rate control ($\hat{P}_c=P_{c}$) independent of SNR whenever SNR$\geq15$, SNR$\geq22$, SNR$\geq25$ dB and SNR$\geq35$ dB for single, two, three and four source scenarios respectively.
\item Test-D maintains $P_f\leq 0.01$ at moderate to high SNRs.
\item The Algorithm also outputs the correct DoA(s) for $99\%$ ($P_c$) of the trials at moderate to high SNRs. 
\end{enumerate}
From the observations, we can conclude that $\pr(B)\to1$ for well separated sources at moderate to high SNRs . We also conclude that Test-D maintains rate control ($\hat{P}_c= P_c$) at moderate to high SNRs for well separated sources, where $\pr(B)\to1$. Whenever the source is detected correctly, the estimation error is zero because of perfect grid matching. We note that although the tests have been performed for $P_c = 0.99$, the rate control for higher $P_c$ was also observed.
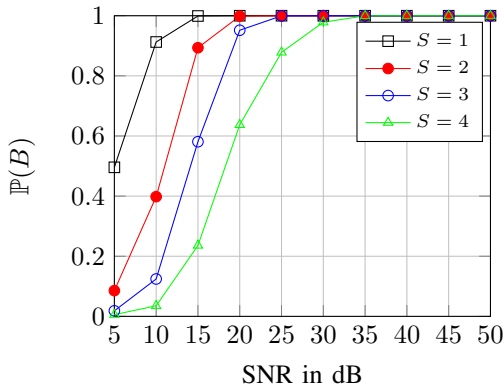
\begin{figure}[h]
\centering
\begin{tikzpicture}
\draw[help lines] (0,0);
\begin{axis}[%
width=5cm,
height=4cm,
scale only axis,
xmin=5,
xmax=50,
xtick distance={5},
grid=both,
grid style={line width=.1pt, draw=gray!10},
major grid style={line width=.2pt,draw=gray!50},
xlabel={SNR in dB},
ymin=0,
ymax=1,
ylabel={$\pr(B)$},
]


\addplot [
color=black,
solid,
mark=square,
mark options={solid}
]
table[row sep=crcr]{
   5.0000     0.4958\\
   10.0000    0.9126\\
   15.0000    0.9991\\
   20.0000    1.0000\\
   25.0000    1.0000\\
   30.0000    1.0000\\
   35.0000    1.0000\\
   40.0000    1.0000\\
   45.0000    1.0000\\
   50.0000    1.0000\\
};
\addlegendentry{\scriptsize{$S=1$}};
\addplot [
color=red,
solid,
mark=*,
mark options={solid}
]
table[row sep=crcr]{
    5.0000    0.0856\\
   10.0000    0.3982\\
   15.0000    0.8935\\
   20.0000    0.9987\\
   25.0000    1.0000\\
   30.0000    1.0000\\
   35.0000    1.0000\\
   40.0000    1.0000\\
   45.0000    1.0000\\
   50.0000    1.0000\\
};
\addlegendentry{\scriptsize{$S=2$}};
\addplot [
color=blue,
solid,
mark=o,
mark options={solid}
]
table[row sep=crcr]{
    5.0000    0.0184\\
   10.0000    0.1244\\
   15.0000    0.5809\\
   20.0000    0.9517\\
   25.0000    0.9996\\
   30.0000    1.0000\\
   35.0000    1.0000\\
   40.0000    1.0000\\
   45.0000    1.0000\\
   50.0000    1.0000\\
};
\addlegendentry{\scriptsize{$S=3$}};
\addplot [
color=green,
solid,
mark=triangle,
mark options={solid}
]
table[row sep=crcr]{
 	5.0000    0.0057\\  
   10.0000    0.0352\\
   15.0000    0.2359\\
   20.0000    0.6371\\
   25.0000    0.8782\\
   30.0000    0.9784\\
   35.0000    0.9997\\
   40.0000    1.0000\\
   45.0000    1.0000\\
   50.0000    1.0000\\
};
\addlegendentry{\scriptsize{$S=4$}};
\end{axis}
\end{tikzpicture}%
\caption{$\pr(B)$ vs SNR for non-orthogonal model}
\label{fig:Fig02}
\end{figure}
\subsection{Grid Matching}
In the simulations for the single source scenario, the source impinges the ULA from the angle $\boldsymbol{\rho}(5)$. Similarly, sources impinge ULA from angles $\boldsymbol{\rho}(3,6)$ for two source scenario. For three and four source scenarios the sources impinge the ULA from $\boldsymbol{\rho}(2,4,6)$ and $\boldsymbol{\rho}(1,3,5,7)$ angles respectively. For $\pr(B)$ simulations, the offset error was maintained at $p_{i} = 0.24r$ for all the sources, where $r$ is the resolution of the grid. Fig. \ref{fig:Fig03} shows the plot of $\pr(B)$ vs SNR. We note that $\pr(B)$ depends on the grid mismatch error $p_i$. For four source scenario and $p_i$ to $0.24r$ for all the sources, $\pr(B)$ is always zero irrespective of SNR (green curve with diamond marker). Reducing $p_i = 0.1r$ for $S=4$ sources, we find that $\pr(B)\to1$ for high SNRs. Table-\ref{Ta7} shows the $\hat{P}_c$ obtained by Algorithm-\ref{Algo3} based on Test-E. The number of sources ($S$) received are indicated in the sub-caption. Fig. \ref{fig:Fig04} shows the square root of the Cramer-Rao bound (SCRB) vs SNR for the model for $S$ sources. For $S=3$ sources, the  root of the mean-square error (RMSE) vs SNR is also plotted with $p_i = 0.24r$ for all sources. The average for the RMSE plot is taken over correctly detected sources which account for $P_c$. Hence, the RMSE is not calculated for SNR $<30$dB where $P_c$ is small. The following observations can be made from Fig. \ref{fig:Fig03}, Fig. \ref{fig:Fig04} and Table-\ref{Ta7}.
\begin{enumerate}
\item Fig. \ref{fig:Fig03} shows that $\pr(B)\to1$ for SNR $>20,35,35$ dB for one, two and three source scenarios. 
\item Test-E does not provide rate control (\emph{i.e.}, $\hat{P}_c<P_{c}$) for SNR$<20$ dB for single source, SNR$<40$ dB for two source and SNR$<40$ dB for three source scenarios respectively. The reason for this behavior is that $P(B)\neq1$ in these scenarios, so the tests fail when SNR is low and sources are not well separated. 
\item The test gives perfect rate control ($\hat{P}_c=P_{c}$) independent of SNR whenever SNR$\geq25$, SNR$\geq40$ and SNR$\geq40$ dB for single, two and three source scenarios.
\item Test-E maintains $P_f\leq 0.01$ at moderate to high SNRs.
\item The Algorithm also outputs the correct DoA(s) for $99\%$ ($P_c$) of the trials at high SNRs.  
\item Fig. \ref{fig:Fig04} shows that the RMSE for $S=3$ is close to the SCRB for SNR$\geq 30$dB.
\end{enumerate}
From the observations, we can conclude that $\pr(B)\to1$ for well separated sources at moderate to high SNRs. We also conclude that Test-E maintains rate control ($\hat{P}_c= P_c$) at high SNRs for well separated sources, where $\pr(B)\to1$. Whenever the source(s) are detected correctly, the estimation error and is close to the SCRB for the grid mismatch model. We note that although the tests have been performed for $P_c = 0.99$, the rate control for higher $P_c$ was also observed.
\begin{figure}
\centering
\begin{tikzpicture}
\draw[help lines] (0,0);
\begin{axis}[%
width=5cm,
height=4cm,
scale only axis,
xmin=5,
xmax=50,
xtick distance={5},
grid=both,
grid style={line width=.1pt, draw=gray!10},
major grid style={line width=.2pt,draw=gray!50},
xlabel={SNR in dB},
ymin=0,
ymax=1,
ylabel={$\pr(B)$},
]


\addplot [
color=black,
solid,
mark=square,
mark options={solid}
]
table[row sep=crcr]{
    5.0000    0.5357\\
   10.0000    0.8957\\
   15.0000    0.9974\\
   20.0000    1.0000\\
   25.0000    1.0000\\
   30.0000    1.0000\\
   35.0000    1.0000\\
   40.0000    1.0000\\
   45.0000    1.0000\\
   50.0000    1.0000\\
};
\addlegendentry{\scriptsize{$S=1$}};
\addplot [
color=red,
solid,
mark=*,
mark options={solid}
]
table[row sep=crcr]{
    5.0000    0.0357\\
   10.0000    0.1326\\
   15.0000    0.3814\\
   20.0000    0.6355\\
   25.0000    0.7756\\
   30.0000    0.8623\\
   35.0000    0.9617\\
   40.0000    0.9989\\
   45.0000    1.0000\\
   50.0000    1.0000\\
};
\addlegendentry{\scriptsize{$S=2$}};
\addplot [
color=blue,
solid,
mark=o,
mark options={solid}
]
table[row sep=crcr]{
    5.0000    0.0222\\
   10.0000    0.0886\\
   15.0000    0.3021\\
   20.0000    0.5534\\
   25.0000    0.7173\\
   30.0000    0.8644\\
   35.0000    0.9706\\
   40.0000    0.9995\\
   45.0000    1.0000\\
   50.0000    1.0000\\
};
\addlegendentry{\scriptsize{$S=3$}};
\addplot [
color=green,
solid,
mark=triangle,
mark options={solid}
]
table[row sep=crcr]{
    5.0000         0\\
   10.0000    0.0001\\
   15.0000    0.0011\\
   20.0000    0.0011\\
   25.0000         0\\
   30.0000         0\\
   35.0000         0\\
   40.0000         0\\
   45.0000         0\\
   50.0000         0\\
};
\addlegendentry{\scriptsize{$S=4$}};
\addplot [
color=brown,
solid,
mark=diamond,
mark options={solid}
]
table[row sep=crcr]{
    5.0000    0.0380\\
   10.0000    0.1351\\
   15.0000    0.4159\\
   20.0000    0.6777\\
   25.0000    0.7742\\
   30.0000    0.8687\\
   35.0000    0.9747\\
   40.0000    0.9999\\
   45.0000    1.0000\\
   50.0000    1.0000\\
};
\addlegendentry{\scriptsize{$S=4$}};
\end{axis}
\end{tikzpicture}%
\caption{$\pr(B)$ for grid matching model}
\label{fig:Fig03}
\end{figure}
\begin{figure}
\centering
\begin{tikzpicture}
\draw[help lines] (0,0);
\begin{axis}[%
width=5cm,
height=5cm,
scale only axis,
xmin=5,
xmax=50,
xtick distance={5},
grid=both,
grid style={line width=.1pt, draw=gray!10},
major grid style={line width=.2pt,draw=gray!50},
xlabel={SNR in dB},
ymin=0,
ymax=0.05,
ylabel={SBCRB},
]


\addplot [
color=black,
solid,
mark=square,
mark options={solid}
]
table[row sep=crcr]{
    5.0000    0.0251\\
   10.0000    0.0141\\
   15.0000    0.0079\\
   20.0000    0.0045\\
   25.0000    0.0025\\
   30.0000    0.0014\\
   35.0000    0.0008\\
   40.0000    0.0004\\
   45.0000    0.0003\\
   50.0000    0.0001\\
};
\addlegendentry{\scriptsize{$S=1$}};
\addplot [
color=red,
solid,
mark=*,
mark options={solid}
]
table[row sep=crcr]{
    5.0000    0.0366\\
   10.0000    0.0274\\
   15.0000    0.0206\\
   20.0000    0.0154\\
   25.0000    0.0116\\
   30.0000    0.0087\\
   35.0000    0.0065\\
   40.0000    0.0049\\
   45.0000    0.0037\\
   50.0000    0.0027\\
};
\addlegendentry{\scriptsize{$S=2$}};
\addplot [
color=blue,
solid,
mark=o,
mark options={solid}
]
table[row sep=crcr]{
    5.0000    0.0439\\
   10.0000    0.0362\\
   15.0000    0.0299\\
   20.0000    0.0247\\
   25.0000    0.0204\\
   30.0000    0.0168\\
   35.0000    0.0139\\
   40.0000    0.0115\\
   45.0000    0.0095\\
   50.0000    0.0078\\
};
\addlegendentry{\scriptsize{$S=3$}};
\addplot [
color=brown,
solid,
mark=diamond,
mark options={solid}
]
table[row sep=crcr]{
   30.0000    0.0274\\
   35.0000    0.0234\\
   40.0000    0.0198\\
   45.0000    0.0178\\
   50.0000    0.0165\\
};
\addlegendentry{\scriptsize{$S=3$},$\sqrt{\text{MSE}}$};
\addplot [
color=green,
solid,
mark=triangle,
mark options={solid}
]
table[row sep=crcr]{
    5.0000    0.0457\\
   10.0000    0.0395\\
   15.0000    0.0342\\
   20.0000    0.0297\\
   25.0000    0.0257\\
   30.0000    0.0222\\
   35.0000    0.0193\\
   40.0000    0.0167\\
   45.0000    0.0144\\
   50.0000    0.0125\\
};
\addlegendentry{\scriptsize{$S=4$}};
\end{axis}
\end{tikzpicture}%
\caption{SBCRB vs SNR for grid matching model}
\label{fig:Fig04}
\end{figure}
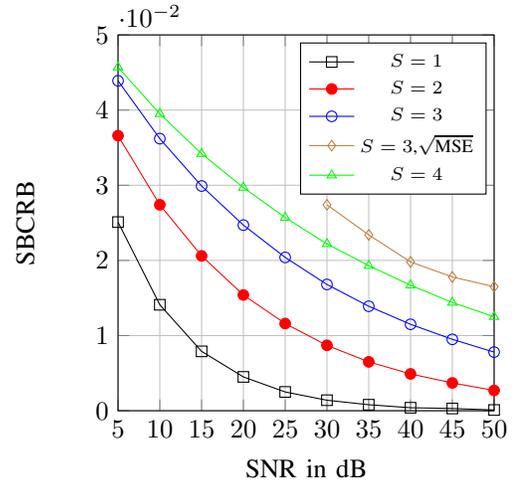
\begin{table*}
\begin{subtable}{0.24\textwidth}
\begin{tabular}{|c|c|c|c|c|}
    \hline
    SNR     & $\hat{P}_c$ & $\hat{P}_f$\\
    \hline
    $5$ dB  & 0.2131 & 0.0179\\
    \hline
    $10$ dB & 0.8220 & 0.0290\\
    \hline
    $15$ dB & 0.9683 & 0.0317\\
    \hline
    $20$ dB & 0.9693 & 0.0307\\
    \hline
    $25$ dB & 0.9694 & 0.0306\\
    \hline
    $30$ dB & 0.9693 & 0.0307\\
    \hline
    $35$ dB & 0.9701 & 0.0299\\
    \hline
    $40$ dB & 0.9696 & 0.0304\\
    \hline
    $45$ dB & 0.9695 & 0.0305\\
    \hline
    $50$ dB & 0.9695 & 0.0305\\
    \hline
\end{tabular}
\subcaption{$S=1$}
\end{subtable}
\begin{subtable}{0.24\textwidth}
\begin{tabular}{|c|c|c|c|c|}
    \hline
    SNR     & $\hat{P}_c$ & $\hat{P}_f$\\
    \hline
    $5$ dB  & 0.2131 & 0.0179\\
    \hline
    $10$ dB & 0.8220 & 0.0290\\
    \hline
    $15$ dB & 0.9683 & 0.0317\\
    \hline
    $20$ dB & 0.9693 & 0.0307\\
    \hline
    $25$ dB & 0.9694 & 0.0306\\
    \hline
    $30$ dB & 0.9693 & 0.0307\\
    \hline
    $35$ dB & 0.9701 & 0.0299\\
    \hline
    $40$ dB & 0.9696 & 0.0304\\
    \hline
    $45$ dB & 0.9695 & 0.0305\\
    \hline
    $50$ dB & 0.9695 & 0.0305\\
    \hline
\end{tabular}
\subcaption{$S=2$}
\end{subtable}
\begin{subtable}{0.24\textwidth}
\begin{tabular}{|c|c|c|c|c|}
    \hline
    SNR     & $\hat{P}_c$ & $\hat{P}_f$\\
    \hline
    $5$ dB  & 0.0003 & 0.0158\\
    \hline
    $10$ dB & 0.0309 & 0.0815\\
    \hline
    $15$ dB & 0.6992 & 0.1759\\
    \hline
    $20$ dB & 0.9658 & 0.0342\\
    \hline
    $25$ dB & 0.9656 & 0.0344\\
    \hline
    $30$ dB & 0.9663 & 0.0337\\
    \hline
    $35$ dB & 0.9659 & 0.0340\\
    \hline
    $40$ dB & 0.9658 & 0.0342\\
    \hline
    $45$ dB & 0.9658 & 0.0342\\
    \hline
    $50$ dB &  0.9660 & 0.0340\\
    \hline
\end{tabular}
\subcaption{$S=3$}
\end{subtable}
\begin{subtable}{0.24\textwidth}
\begin{tabular}{|c|c|c|c|c|}
    \hline
    SNR     & $\hat{P}_c$ & $\hat{P}_f$\\
    \hline
    $5$ dB  & 0 & 0.0077\\
    \hline
    $10$ dB & 0.0042& 0.0041\\
    \hline
    $15$ dB & 0.3821 & 0.0233\\
    \hline
    $20$ dB & 0.9616 & 0.0379\\
    \hline
    $25$ dB & 0.9631 & 0.0369\\
    \hline
    $30$ dB & 0.9624 & 0.0376\\
    \hline
    $35$ dB & 0.9632 & 0.0368\\
    \hline
    $40$ dB & 0.9620 & 0.0380\\
    \hline
    $45$ dB & 0.9620 & 0.0380\\
    \hline
    $50$ dB & 0.9621 & 0.0379\\
    \hline
\end{tabular}
\subcaption{$S=4$}
\end{subtable}
\caption{Asymptotic Covariance Test}
\label{Ta1}
\end{table*}
\begin{table*}
\begin{subtable}{0.24\textwidth}
\begin{tabular}{|c|c|c|c|c|}
    \hline
    SNR     & $\hat{P}_c$ & $\hat{P}_f$\\
    \hline
    $5$ dB  & 0.1382 & 0.0046\\
    \hline
    $10$ dB & 0.7557 & 0.0085\\
    \hline
    $15$ dB & 0.9903 & 0.0097\\
    \hline
    $20$ dB & 0.9907 & 0.0093\\
    \hline
    $25$ dB & 0.9902 & 0.0098\\
    \hline
    $30$ dB & 0.9907 & 0.0093\\
    \hline
    $35$ dB & 0.9902 & 0.0098\\
    \hline
    $40$ dB & 0.9901 & 0.0099\\
    \hline
    $45$ dB & 0.9906 & 0.0094\\
    \hline
    $50$ dB & 0.9901 & 0.0099\\
    \hline
  \end{tabular}
\subcaption{$S = 1$}
\end{subtable}
\begin{subtable}{0.24\textwidth}
\begin{tabular}{|c|c|c|c|c|}
    \hline
    SNR     & $\hat{P}_c$ & $\hat{P}_f$\\
    \hline
    $5$ dB  & 0.0022 & 0.0026\\
    \hline
    $10$ dB & 0.1164 & 0.0025\\
    \hline
    $15$ dB & 0.9233 & 0.0089\\
    \hline
    $20$ dB & 0.9909 & 0.0091\\
    \hline
    $25$ dB & 0.9906 & 0.0094\\
    \hline
    $30$ dB & 0.9902 & 0.0098\\
    \hline
    $35$ dB & 0.9904 & 0.0096\\
    \hline
    $40$ dB & 0.9911 & 0.0089\\
    \hline
    $45$ dB & 0.9907 & 0.0093\\
    \hline
    $50$ dB & 0.9908 & 0.0092\\
    \hline
\end{tabular}
\subcaption{$S = 2$}
\end{subtable}
\begin{subtable}{0.24\textwidth}
\begin{tabular}{|c|c|c|c|c|}
    \hline
    SNR     & $\hat{P}_c$ & $\hat{P}_f$\\
    \hline
    $5$ dB  & 0 & 0.0028\\
    \hline
    $10$ dB & 0.0089 & 0.0012\\
    \hline
    $15$ dB & 0.5525 & 0.0074\\
    \hline
    $20$ dB & 0.9899 & 0.0100\\
    \hline
    $25$ dB & 0.9898 & 0.0103\\
    \hline
    $30$ dB & 0.9900 & 0.0100\\
    \hline
    $35$ dB & 0.9902 & 0.0098\\
    \hline
    $40$ dB & 0.9902 & 0.0098\\
    \hline
    $45$ dB & 0.9905 & 0.0095\\
    \hline
    $50$ dB & 0.9904 & 0.0096\\
    \hline
\end{tabular}
\subcaption{$S = 3$}
\end{subtable}
\begin{subtable}{0.24\textwidth}
\begin{tabular}{|c|c|c|c|c|}
    \hline
    SNR     & $\hat{P}_c$ & $\hat{P}_f$\\
    \hline
    $5$ dB  & 0      & 0.0022\\
    \hline
    $10$ dB & 0      & 0\\
    \hline
    $15$ dB & 0.2118 & 0.0043\\
    \hline
    $20$ dB & 0.9868 & 0.0107\\
    \hline
    $25$ dB & 0.9891 & 0.0109\\
    \hline
    $30$ dB & 0.9890 & 0.0110\\
    \hline
    $35$ dB & 0.9897 & 0.0103\\
    \hline
    $40$ dB & 0.9892 & 0.0108\\
    \hline
    $45$ dB & 0.9891 & 0.0109\\
    \hline
    $50$ dB & 0.9898 & 0.0102\\
    \hline
\end{tabular}
\subcaption{$S = 4$}
\end{subtable}
\caption{Exact Covariance Test}
\end{table*}
\begin{table*}
\begin{subtable}{0.24\textwidth}
\begin{tabular}{|c|c|c|c|c|}
    \hline
    SNR     & $\hat{P}_c$ & $\hat{P}_f$\\
    \hline
    $5$ dB  & 0.1629 & 0.0092\\
    \hline
    $10$ dB & 0.8168 & 0.0097\\
    \hline
    $15$ dB & 0.9901 & 0.0099\\
    \hline
    $20$ dB & 0.9906 & 0.0094\\
    \hline
    $25$ dB & 0.9903 & 0.0097\\
    \hline
    $30$ dB & 0.9897 & 0.0103\\
    \hline
    $35$ dB & 0.9903 & 0.0097\\
    \hline
    $40$ dB & 0.9898 & 0.0102\\
    \hline
    $45$ dB & 0.9901 & 0.0098\\
    \hline
    $50$ dB & 0.9903 & 0.0097\\
    \hline
  \end{tabular}
\subcaption{$S = 1$}
\end{subtable}
\begin{subtable}{0.24\textwidth}
\begin{tabular}{|c|c|c|c|c|}
    \hline
    SNR     & $\hat{P}_c$ & $\hat{P}_f$\\
    \hline
    $5$ dB  & 0.0025 & 0.0074\\
    \hline
    $10$ dB & 0.1424 & 0.0065\\
    \hline
    $15$ dB & 0.9564 & 0.0097\\
    \hline
    $20$ dB & 0.9896 & 0.0104\\
    \hline
    $25$ dB & 0.9902 & 0.0098\\
    \hline
    $30$ dB & 0.9905 & 0.0095\\
    \hline
    $35$ dB & 0.9905 & 0.0095\\
    \hline
    $40$ dB & 0.9901 & 0.0099\\
    \hline
    $45$ dB & 0.9904 & 0.0096\\
    \hline
    $50$ dB & 0.9895 & 0.0106\\
    \hline
\end{tabular}
\subcaption{$S = 2$}
\end{subtable}
\begin{subtable}{0.24\textwidth}
\begin{tabular}{|c|c|c|c|c|}
    \hline
    SNR     & $\hat{P}_c$ & $\hat{P}_f$\\
    \hline
    $5$ dB  & 0 & 0.0061\\
    \hline
    $10$ dB & 0.0083 & 0.0068\\
    \hline
    $15$ dB & 0.6575 & 0.0097\\
    \hline
    $20$ dB & 0.9896 & 0.0104\\
    \hline
    $25$ dB & 0.9904 & 0.0096\\
    \hline
    $30$ dB & 0.9895 & 0.0105\\
    \hline
    $35$ dB & 0.9897 & 0.0103\\
    \hline
    $40$ dB & 0.9908 & 0.0092\\
    \hline
    $45$ dB & 0.9902 & 0.0098\\
    \hline
    $50$ dB & 0.9896 & 0.0104\\
    \hline
\end{tabular}
\subcaption{$S = 3$}
\end{subtable}
\begin{subtable}{0.24\textwidth}
\begin{tabular}{|c|c|c|c|c|}
    \hline
    SNR     & $\hat{P}_c$ & $\hat{P}_f$\\
    \hline
    $5$ dB  & 0      & 0.0051\\
    \hline
    $10$ dB & 0      & 0.0053\\
    \hline
    $15$ dB & 0.2669 & 0.0075\\
    \hline
    $20$ dB & 0.9892 & 0.0101\\
    \hline
    $25$ dB & 0.9900 & 0.0100\\
    \hline
    $30$ dB & 0.9896 & 0.0104\\
    \hline
    $35$ dB & 0.9895 & 0.0105\\
    \hline
    $40$ dB & 0.9901 & 0.0098\\
    \hline
    $45$ dB & 0.9905 & 0.0095\\
    \hline
    $50$ dB & 0.9896 & 0.0104\\
    \hline
\end{tabular}
\subcaption{$S = 4$}
\end{subtable}
\caption{Test-A}
\end{table*}
\begin{table*}
\begin{subtable}{0.24\textwidth}
\begin{tabular}{|c|c|c|c|c|}
    \hline
    SNR     & $\hat{P}_c$ & $\hat{P}_f$\\
    \hline
    $5$ dB  & 0.1475 & 0.0050\\
    \hline
    $10$ dB & 0.7775 & 0.0089\\
    \hline
    $15$ dB & 0.9900 & 0.0100\\
    \hline
    $20$ dB & 0.9896 & 0.0104\\
    \hline
    $25$ dB & 0.9902 & 0.0098\\
    \hline
    $30$ dB & 0.9897 & 0.0103\\
    \hline
    $35$ dB & 0.9901 & 0.0099\\
    \hline
    $40$ dB & 0.9900 & 0.0100\\
    \hline
    $45$ dB & 0.9900 & 0.0100\\
    \hline
    $50$ dB & 0.9901 & 0.0099\\
    \hline
  \end{tabular}
\subcaption{$S = 1$}
\end{subtable}
\begin{subtable}{0.24\textwidth}
\begin{tabular}{|c|c|c|c|c|}
    \hline
    SNR     & $\hat{P}_c$ & $\hat{P}_f$\\
    \hline
    $5$ dB  & 0.0027 & 0.0039\\
    \hline
    $10$ dB & 0.1249 & 0.0030\\
    \hline
    $15$ dB & 0.9364 & 0.0094\\
    \hline
    $20$ dB & 0.9904 & 0.0094\\
    \hline
    $25$ dB & 0.9899 & 0.0101\\
    \hline
    $30$ dB & 0.9899 & 0.0101\\
    \hline
    $35$ dB & 0.9898 & 0.0102\\
    \hline
    $40$ dB & 0.9902 & 0.0098\\
    \hline
    $45$ dB & 0.9900 & 0.0100\\
    \hline
    $50$ dB & 0.9901 & 0.0099\\
    \hline
\end{tabular}
\subcaption{$S = 2$}
\end{subtable}
\begin{subtable}{0.24\textwidth}
\begin{tabular}{|c|c|c|c|c|}
    \hline
    SNR     & $\hat{P}_c$ & $\hat{P}_f$\\
    \hline
    $5$ dB  & 0      & 0.0033\\
    \hline
    $10$ dB & 0.0083 & 0.0014\\
    \hline
    $15$ dB & 0.5975 & 0.0074\\
    \hline
    $20$ dB & 0.9900 & 0.0100\\
    \hline
    $25$ dB & 0.9900 & 0.0100\\
    \hline
    $30$ dB & 0.9897 & 0.0103\\
    \hline
    $35$ dB & 0.9901 & 0.0099\\
    \hline
    $40$ dB & 0.9901 & 0.0098\\
    \hline
    $45$ dB & 0.9897 & 0.0103\\
    \hline
    $50$ dB & 0.9895 & 0.0105\\
    \hline
\end{tabular}
\subcaption{$S = 3$}
\end{subtable}
\begin{subtable}{0.24\textwidth}
\begin{tabular}{|c|c|c|c|c|}
    \hline
    SNR     & $\hat{P}_c$ & $\hat{P}_f$\\
    \hline
    $5$ dB  & 0      & 0\\
    \hline
    $10$ dB & 0      & 0\\
    \hline
    $15$ dB & 0.0707 & 0.0002\\
    \hline
    $20$ dB & 0.9905 & 0.0016\\
    \hline
    $25$ dB & 0.9900 & 0.0100\\
    \hline
    $30$ dB & 0.9897 & 0.0103\\
    \hline
    $35$ dB & 0.9901 & 0.0099\\
    \hline
    $40$ dB & 0.9901 & 0.0098\\
    \hline
    $45$ dB & 0.9897 & 0.0103\\
    \hline
    $50$ dB & 0.9895 & 0.0105\\
    \hline
\end{tabular}
\subcaption{$S = 4$}
\end{subtable}
\caption{Test-B}
\end{table*}
\begin{table*}
\begin{subtable}{0.24\textwidth}
\begin{tabular}{|c|c|c|c|c|}
    \hline
    SNR     & $\hat{P}_c$ & $\hat{P}_f$\\
    \hline
    $5$ dB  & 0.0254 & 0.0103\\
    \hline
    $10$ dB & 0.1345 & 0.0098\\
    \hline
    $15$ dB & 0.5467 & 0.0100\\
    \hline
    $20$ dB & 0.9653 & 0.0110\\
    \hline
    $25$ dB & 0.9896 & 0.0104\\
    \hline
    $30$ dB & 0.9897 & 0.0103\\
    \hline
    $35$ dB & 0.9901 & 0.0099\\
    \hline
    $40$ dB & 0.9899 & 0.0101\\
    \hline
    $45$ dB & 0.9905 & 0.0095\\
    \hline
    $50$ dB & 0.9896 & 0.0104\\
    \hline
  \end{tabular}
\subcaption{$S = 1$}
\end{subtable}
\begin{subtable}{0.24\textwidth}
\begin{tabular}{|c|c|c|c|c|}
    \hline
    SNR     & $\hat{P}_c$ & $\hat{P}_f$\\
    \hline
    $5$ dB  & 0.0024 & 0.0075\\
    \hline
    $10$ dB & 0.0212 & 0.0071\\
    \hline
    $15$ dB & 0.1660 & 0.0071\\
    \hline
    $20$ dB & 0.6900 & 0.0084\\
    \hline
    $25$ dB & 0.9886 & 0.0078\\
    \hline
    $30$ dB & 0.9920 & 0.0080\\
    \hline
    $35$ dB & 0.9928 & 0.0072\\
    \hline
    $40$ dB & 0.9923 & 0.0077\\
    \hline
    $45$ dB & 0.9920 & 0.0080\\
    \hline
    $50$ dB & 0.9926 & 0.0074\\
    \hline
\end{tabular}
\subcaption{$S = 2$}
\end{subtable}
\begin{subtable}{0.24\textwidth}
\begin{tabular}{|c|c|c|c|c|}
    \hline
    SNR     & $\hat{P}_c$ & $\hat{P}_f$\\
    \hline
    $5$ dB  & 0      & 0.0196\\
    \hline
    $10$ dB & 0.0085 & 0.0299\\
    \hline
    $15$ dB & 0.0956 & 0.0873\\
    \hline
    $20$ dB & 0.5469 & 0.1236\\
    \hline
    $25$ dB & 0.9754 & 0.0176\\
    \hline
    $30$ dB & 0.9900 & 0.0100\\
    \hline
    $35$ dB & 0.9897 & 0.0103\\
    \hline
    $40$ dB & 0.9898 & 0.0102\\
    \hline
    $45$ dB & 0.9899 & 0.0101\\
    \hline
    $50$ dB & 0.9900 & 0.0100\\
    \hline
\end{tabular}
\subcaption{$S = 3$}
\end{subtable}
\begin{subtable}{0.24\textwidth}
\begin{tabular}{|c|c|c|c|c|}
    \hline
    SNR     & $\hat{P}_c$ & $\hat{P}_f$\\
    \hline
    $5$ dB  & 0      & 0.0248\\
    \hline
    $10$ dB & 0.0039 & 0.0120\\
    \hline
    $15$ dB & 0.0544 & 0.0092\\
    \hline
    $20$ dB & 0.4096 & 0.0096\\
    \hline
    $25$ dB & 0.9503 & 0.0092\\
    \hline
    $30$ dB & 0.9908 & 0.0092\\
    \hline
    $35$ dB & 0.9907 & 0.0094\\
    \hline
    $40$ dB & 0.9911 & 0.0089\\
    \hline
    $45$ dB & 0.9909 & 0.0091\\
    \hline
    $50$ dB & 0.9906 & 0.0094\\
    \hline
\end{tabular}
\subcaption{$S = 4$}
\end{subtable}
\caption{Test-C}
\label{Ta5}
\end{table*}
\begin{table*}
\begin{subtable}{0.24\textwidth}
\begin{tabular}{|c|c|c|c|c|}
    \hline
    SNR     & $\hat{P}_c$ & $\hat{P}_f$\\
    \hline
    $5$ dB  & 0.1337 & 0.0388\\
    \hline
    $10$ dB & 0.7653 & 0.0502\\
    \hline
    $15$ dB & 0.9893 & 0.0107\\
    \hline
    $20$ dB & 0.9894 & 0.0106\\
    \hline
    $25$ dB & 0.9903 & 0.0097\\
    \hline
    $30$ dB & 0.9900 & 0.0100\\
    \hline
    $35$ dB & 0.9900 & 0.0100\\
    \hline
    $40$ dB & 0.9898 & 0.0102\\
    \hline
    $45$ dB & 0.9902 & 0.0098\\
    \hline
    $50$ dB & 0.9901 & 0.0099\\
    \hline
  \end{tabular}
\subcaption{$S = 1$}
\end{subtable}
\begin{subtable}{0.24\textwidth}
\begin{tabular}{|c|c|c|c|c|}
    \hline
    SNR     & $\hat{P}_c$ & $\hat{P}_f$\\
    \hline
    $5$ dB  & 0      & 0.0473\\
    \hline
    $10$ dB & 0.0709 & 0.1228\\
    \hline
    $15$ dB & 0.8419 & 0.0629\\
    \hline
    $20$ dB & 0.9885 & 0.0114\\
    \hline
    $25$ dB & 0.9898 & 0.0102\\
    \hline
    $30$ dB & 0.9899 & 0.0101\\
    \hline
    $35$ dB & 0.9901 & 0.0099\\
    \hline
    $40$ dB & 0.9899 & 0.0101\\
    \hline
    $45$ dB & 0.9893 & 0.0107\\
    \hline
    $50$ dB & 0.9901 & 0.0099\\
    \hline
\end{tabular}
\subcaption{$S = 2$}
\end{subtable}
\begin{subtable}{0.24\textwidth}
\begin{tabular}{|c|c|c|c|c|}
    \hline
    SNR     & $\hat{P}_c$ & $\hat{P}_f$\\
    \hline
    $5$ dB  & 0      & 0.0612\\
    \hline
    $10$ dB & 0.0024 & 0.2130\\
    \hline
    $15$ dB & 0.3574 & 0.5186\\
    \hline
    $20$ dB & 0.9435 & 0.0523\\
    \hline
    $25$ dB & 0.9893 & 0.0107\\
    \hline
    $30$ dB & 0.9901 & 0.0099\\
    \hline
    $35$ dB & 0.9895 & 0.0105\\
    \hline
    $40$ dB & 0.9899 & 0.0101\\
    \hline
    $45$ dB & 0.9893 & 0.0107\\
    \hline
    $50$ dB & 0.9899 & 0.0101\\
    \hline
\end{tabular}
\subcaption{$S = 3$}
\end{subtable}
\begin{subtable}{0.24\textwidth}
\begin{tabular}{|c|c|c|c|c|}
    \hline
    SNR     & $\hat{P}_c$ & $\hat{P}_f$\\
    \hline
    $5$ dB  & 0      & 0.0828\\
    \hline
    $10$ dB & 0      & 0.2317\\
    \hline
    $15$ dB & 0.0602 & 0.5768\\
    \hline
    $20$ dB & 0.6308 & 0.3680\\
    \hline
    $25$ dB & 0.8715 & 0.1285\\
    \hline
    $30$ dB & 0.9701 & 0.0299\\
    \hline
    $35$ dB & 0.9903 & 0.0097\\
    \hline
    $40$ dB & 0.9906 & 0.0094\\
    \hline
    $45$ dB & 0.9904 & 0.0096\\
    \hline
    $50$ dB & 0.9908 & 0.0092\\
    \hline
\end{tabular}
\subcaption{$S = 4$}
\end{subtable}
\caption{Test-D}
\label{Ta6}
\end{table*}
\begin{table*}
\centering
\begin{subtable}{0.24\textwidth}
\begin{tabular}{|c|c|c|c|c|}
    \hline
    SNR     & $\hat{P}_c$ & $\hat{P}_f$\\
    \hline
    $5$ dB  & 0      & 0\\
    \hline
    $10$ dB & 0.0112 & 0.0105\\
    \hline
    $15$ dB & 0.4334 & 0.3841\\
    \hline
    $20$ dB & 0.8773 & 0.0980\\
    \hline
    $25$ dB & 0.9902 & 0.0098\\
    \hline
    $30$ dB & 0.9890 & 0.0101\\
    \hline
    $35$ dB & 0.9905 & 0.0095\\
    \hline
    $40$ dB & 0.9900 & 0.0100\\
    \hline
    $45$ dB & 0.9907 & 0.0093\\
    \hline
    $50$ dB & 0.9908 & 0.0092\\
    \hline
  \end{tabular}
\subcaption{$S = 1$}
\end{subtable}
\begin{subtable}{0.24\textwidth}
\begin{tabular}{|c|c|c|c|c|}
    \hline
    SNR     & $\hat{P}_c$ & $\hat{P}_f$\\
    \hline
    $5$ dB  & 0      & 0.0012\\
    \hline
    $10$ dB & 0      & 0.0340\\
    \hline
    $15$ dB & 0.1365 & 0.4188\\
    \hline
    $20$ dB & 0.6156 & 0.3751\\
    \hline
    $25$ dB & 0.7706 & 0.2294\\
    \hline
    $30$ dB & 0.8634 & 0.1366\\
    \hline
    $35$ dB & 0.9626 & 0.0374\\
    \hline
    $40$ dB & 0.9900 & 0.0100\\
    \hline
    $45$ dB & 0.9899 & 0.0101\\
    \hline
    $50$ dB & 0.9903 & 0.0097\\
    \hline
\end{tabular}
\subcaption{$S = 2$}
\end{subtable}
\begin{subtable}{0.24\textwidth}
\begin{tabular}{|c|c|c|c|c|}
    \hline
    SNR     & $\hat{P}_c$ & $\hat{P}_f$\\
    \hline
    $5$ dB  & 0      & 0.0085\\
    \hline
    $10$ dB & 0      & 0.1071\\
    \hline
    $15$ dB & 0.0694 & 0.5859\\
    \hline
    $20$ dB & 0.5222 & 0.4733\\
    \hline
    $25$ dB & 0.7118 & 0.2878\\
    \hline
    $30$ dB & 0.8711 & 0.1289\\
    \hline
    $35$ dB & 0.9691 & 0.0309\\
    \hline
    $40$ dB & 0.9898 & 0.0102\\
    \hline
    $45$ dB & 0.9902 & 0.0098\\
    \hline
    $50$ dB & 0.9898 & 0.0102\\
    \hline
\end{tabular}
\subcaption{$S = 3$}
\end{subtable}
\caption{Test-E}
\label{Ta7}
\end{table*}
\section{Conclusions}\label{concl}
In this work, we propose tests for joint detection and estimation of multiple targets using single snapshot measurements at moderate to high SNR. These tests can also be interpreted as stopping criterion for homotopy based (group) lasso estimators, since they provide a stopping criteria as the (group) lasso estimator travels the (group) lasso path. The proposed algorithms offer control over the probability of correct detection of the sources by choosing the appropriate threshold. Although we have applied the algorithm only for the DoA problem, the algorithm can be used for any linear model with Gaussian noise problem. 
\bibliographystyle{IEEEtran}
\bibliography{IEEEabrv,bibone}
\section{Appendix}
\subsection{Proof of Theorem-\ref{Finitesamplecov}}\label{A1}
In the moderate SNR regime, $\mathsmaller{\frac{\tau_{j}} {\sigma}, j=S+1,S+2,\ldots,M}$ are the order statistics of Rayleigh random variable with p.d.f $f(x)$ and c.d.f $F(x) = 1-\exp(-x^2/2)$. Defining $M-S = n$ and $V_j = \tau_{S+j}/\sigma$, we have $\mathsmaller{V_n\leq\ldots\leq V_{j}\leq\ldots V_{1}}$. Defining $V_{j} = X_{n+1-i}$, we have $\mathsmaller{X_1\leq\ldots X_{i}\leq\ldots\leq X_n}$.

We first require the joint pdf of $V_1,V_2$ or $X_n,X_{n-1}$. The joint pdf of consecutive order statistics is \cite[Chapter-2]{arnold}
\begin{equation*}
f_{X_{k},X_{k+1}}(x,y) = C_0\{F(x)\}^{k-1}\{1-F(y)\}^{n-k-1}f(x)f(y),
\end{equation*}
where $C_0 = \frac{n!}{(k-1)!(n-k-1)!}$. Substituting $k =n-1$,
\begin{equation*}
f_{X_{n-1},X_{n}}(x,y) = C\{F(x)\}^{n-2}f(x)f(y),0<x<y<\infty,
\end{equation*}
where $\mathsmaller{C = \frac{n!}{(n-2)!}}$. The joint pdf of $X_{n}$ and $\mathsmaller{W = X_{n}-X_{n-1}}$ is,
\begin{equation*}
\resizebox{\hsize}{!}{$f_{W,X_n}(w,y) = C\{F(y-w)\}^{n-2}f(y-w)f(y),0<w<y<\infty$}.
\end{equation*}
Now, the joint p.d.f of $X_n$ and $T_{S+1} = X_nW$ is,
\begin{equation*}
\resizebox{\hsize}{!}{$f_{T_{S+1},X_n}(t,y) = C\{F(y-t/y)\}^{n-2}f(y-t/y)f(y)\frac{1}{y},0<t<y^2<\infty$}.
\end{equation*}
Finally the p.d.f of $T_{S+1}$ is obtained by integration of the above equation w.r.t. $y$. Hence,
\begin{equation*}
f_{T_{S+1}}(t) = \int_{\sqrt{t}}^{\infty}C\{F(y-t/y)\}^{n-2}f(y-t/y)f(y)\frac{1}{y}dy.
\end{equation*}
Now the cdf of the co-variance test statistics is,
\begin{align*}
&F_{T_{S+1}}(\eta) = \INT{0}{\eta}{\INT{\sqrt{t}}{\infty}{C\{F(y-t/y)\}^{n-2}f(y-t/y)f(y)\frac{1}{y}dydt}},\\
&=\INT{0}{\sqrt{\eta}}{\INT{0}{\hspace{2mm}y^2}{C\{F(y-t/y)\}^{n-2}f(y-t/y)f(y)\frac{1}{y}dtdy}},\\
&+\INT{\sqrt{\eta}}{\infty}{\INT{0}{\eta}{C\{F(y-t/y)\}^{n-2}f(y-t/y)f(y)\frac{1}{y}dtdy}},\\
&= 1-n\INT{\sqrt{\eta}}{\infty}{y\exp(-y^2/2)\{1-\exp\frac{-(y-\eta/y)^2}{2}\}^{n-1}dy}.
\end{align*}
\subsection{Proof of Theorem-\ref{oversamp}}\label{A3}
Assuming event $B$ is true and $S$ sources. Let $J = \{j_{1},\ldots,j_{S}\}$ be the the active set, after $S$ knot points. Now, at the $(S+1)^{st}$ knot point, $\tau_{S+1} = \max\limits_{k\notin J}\Lambda_{k}$, $\Lambda_{k} = |\a_{k}^{H}(\b-\A_{J}\xhat_{J})|$ for some $k\in J^{c}$ and $\xhat_{J}$ satisfies $\Lambda_{k}\o =|\A_{J}^{H}(\b-\A_{J}\xhat_{J})|,\hspace{1mm}k\in J^{c}$. Hence, we obtain the following set of $|J|$ equations for $\xhat$
\begin{equation}
|\a_{k}^{H}(\b-\A_{J}\xhat_{J})| = |\a_{j_{i}}^{H}(\b-\A_{J}\xhat_{J})|\hspace{1mm}\forall j_{i}\in J, k\in J^{c}.
\end{equation}
Solving for $\xhat$ from the above equations and substituting back in the expression for $\Lambda_{r}$, we obtain
$\Lambda_{r} = |\a_{r}^{H}\Q_{M-S}\v|,r\in J^{c}$, where $\Q_{M-S}$ is a projection matrix with $S$ zero eigen values. Since, $\v$ is a complex Gaussian random variable with zero mean and variance $\sigma^2$, each $\Lambda_{r}^{2}/\sigma^{2}$ are correlated $\chi^{2}$ random variables. Hence, the test $D_{1}$ is a maximum of correlated $\chi^{2}$ random variables whose c.d.f is given by, 
\begin{flalign*}
&=F_{D_{1}}(\eta) = \pr(D_{1}\leq\eta)=\pr\left( \max_{r\in J^{c}}\Lambda^{2}_{r}/\sigma^2\leq\eta\right),&\\
&=\pr(\Lambda_{1}^{2}/\sigma^2\leq\eta,\ldots,\Lambda_{|J|^{c}}^{2}/\sigma^2\leq\eta)
=\INT{0}{\eta}{f_{\u}(\u)\mathrm{d}(\u)},&\\
&\stackrel{(a)}=\INT{0}{\infty}{f_{\u}(\u)\mathbb{I}(\u,\eta)d(\u)}\stackrel{(b)}=\INT{0}{\infty}{\hat{f}_{\z}(\z)\PROD{i}{1}{M-S}{\frac{(1-e^{-j\eta z_{i}})}{jz_{i}}}d\z},&\\
&=\INT{0}{\infty}{(\mathrm{det}(\I-j\mathrm{Diag}(\z)\R_{\mathsmaller{M-S}}))^{-1}\PROD{i}{1}{M-S}{\frac{(1-e^{-j\eta z_{i}})}{jz_{i}}}d\z},&\\
&=\PROD{i}{1}{M-S}{\INT{z_i=0}{\infty}{\hspace{2mm}\frac{(1-e^{-j\eta z_{i}})}{jz_{i}(1-j\varrho_{i}z_{i})}}dz_i}
=\PROD{i}{1}{M-S}{(1-e^{-\eta/\varrho_{i}})}.&
\end{flalign*}
In the above equations, $\mathbb{I}(\u,\eta)=[0,\eta]^{M-S}$, $f_{\u}$ denotes the joint p.d.f of $\Lambda_{r},r\in J^{c}$ in (a) and is degenerate because $\Q_{\mathsmaller{M-S}}$ is singular. Hence, we use the Parseval theorem to obtain (b) and then the characteristic function of correlated $\chi^2$ random variables from \cite{mvray} to evaluate the c.d.f. 
\subsection{Proof of Theorem-\ref{testGMM}}\label{A4}
Assuming event $B$ is true and $S$ sources. Let $J = \{j_{1},\ldots,j_{S}\}$ be the the active group, after $S$ knot points. Now, at the $(S+1)^{st}$ knot point, $\tau_{S+1} = \max\limits_{k\notin J}\Lambda_{k}$, $\Lambda_{k} = \|\P_{k}^{H}(\overline{\b}-\P_{j_{i}}\yhat_{J})\|_{2}$ for all $j_{i}\in J$ and for some $k\in J^{c}$ and $\yhat_{J}$ satisfies $\Lambda_{k} =\|\P_{J}^{H}(\overline{\b}-\P_{J}\yhat_{J})\|_{2}$, for the chosen $k$. Hence, we obtain the following set of equations for $\yhat$
\begin{equation*}
\|\P_{k}^{H}(\overline{\b}-\P_{J}\yhat_{J})\|_{2} = \|\P_{j_{i}}^{H}(\overline{\b}-\A_{J}\yhat_{J})\|_{2}\hspace{1mm}\forall j_{i}\in J, k\in J^{c}.
\end{equation*}
Solving for $\yhat$ from the above equations and substituting back in the expression for $\Lambda_{r}$, we obtain
$\Lambda_{r} = \|\P_{r}^{H}\Q_{2(M-S)}\v\|,r\in J^{c}$, where $\Q_{2(M-S)}$ is a projection matrix with $2S$ zero eigen values. Since, $\v$ is a complex Gaussian random variable with zero mean and variance $\sigma^2$, each $\Lambda_{r}^{2}/\sigma^{2}$ is a correlated $\chi^{2}$ random variable. Hence, the test $E_{1}$ is a maximum of correlated $\chi^{2}$ random variables whose c.d.f is given by, 
\begin{flalign*}
&=F_{E_{1}}(\eta) = \pr(E_{1}\leq\eta)=\pr(\max_{r\in J^{c}}\Lambda^{2}_{r}/\sigma^2\leq\eta),&\\
&=\pr(\Lambda_{1}^{2}/\sigma^2\leq\eta,\ldots,\Lambda_{|J|^{c}}^{2}/\sigma^2\leq\eta)
=\INT{0}{\eta}{f_{\u}(\u)\mathrm{d}(\u)},&\\
&\stackrel{(a)}=\INT{0}{\infty}{f_{\u}(\u)\mathbb{I}(\u,\eta)d(\u)}\stackrel{(b)}=\INT{0}{\infty}{\hat{f}_{\z}(\z)\PROD{i}{1}{M-S}{\frac{(1-e^{-j\eta z_{i}})}{jz_{i}}}d\z},&\\
&=\INT{0}{\infty}{(\mathrm{det}(\I-j\mathrm{Diag}(\z)\R_{\mathsmaller{M-S}})^{-1}(\mathrm{det}(\I-j\mathrm{Diag}(\z)\T_{\mathsmaller{M-S}})^{-1}}&\\
&\hspace{1cm}\times\PROD{i}{1}{M-S}{\frac{(1-e^{-j\eta z_{i}})}{jz_{i}}}d\z &\\
&=\PROD{i}{1}{M-S}{\INT{z_i=0}{\infty}{\hspace{2mm}\frac{(1-e^{-j\eta z_{i}})}{jz_{i}(1-j\varrho_{i}z_{i})(1-j\varepsilon_{i}z_{i})}}dz_i}&
\end{flalign*}
In the above equations $\mathbb{I}(\u,\eta)$ denotes a unit box from $0$ to $\eta$, $f_{\u}$ denotes the joint p.d.f of $\Lambda_{r},r\in J^{c}$ in (a) and is degenerate because $\Q_{\mathsmaller{M-S}}$ is singular. Hence, we use the Parseval theorem to obtsin (b) and then the characteristic function of correlated $\chi^2$ random variables from \cite{mvray} to evaluate the c.d.f.
\subsection{Proof of Theorem-\ref{asympcov}}\label{A2}
We note that Rayleigh random variables ($V_{i}$) satisfy the Von-Mises condition. Hence $\exists$ constants $a_{M} = F^{-1}(1-1/M) = \sqrt{2\log(M)}$ and $b_M = pF^{\prime}(a_{M}) = \sqrt{2\log(M)}$ s.t. $b_{M}(\frac{V_1}{\sigma}-a_{M})\xrightarrow{d}-\log(E_{0})$, where $-\log E_{0}$ has type I extreme value distribution \cite{siglass,dehaan}. From \cite{weiss}, for any fixed $l\geq 1$, the random variables $ W_{0} = b_{M}(\frac{V_{l+1}} {\sigma}-a_{M})$ and $W_{i} = b_{M}(\frac{(V_i-V_{i+1})}{\sigma}), i = 1,\ldots,l$ converge jointly as $(W_0,W_1,W_2,\ldots,W_l)\xrightarrow{d}$ $(\log G_{0},E_{1}/1,E_{2}/2,\ldots,E_{l}/l)$, where $G_{0},E_{1},\ldots,E_{l}$ are independent and $G_{0}$ is Gamma distributed with scale parameter $1$ and shape parameter $l$, and $E_{1},\ldots,E_{l}$ are standard exponentials. We have,
\begin{align*}
T_{\mathsmaller{S+k}}&= \frac{V_k}{\sigma^2}(V_{k}-V_{k+1}) = \Bigg(a_{M}+\frac{W_0}{b_{M}}+\SUM{j}{k}{l}
{\hspace{1mm}\frac{W_{j}}{b_{M}}}\Bigg)\frac{W_{k}}{b_{M}},\\
&=\frac{a_{M}}{b_{M}}W_{k} +\frac{1}{b^{2}_{M}}\Bigg(W_0+\SUM{j}{k}{l}{\hspace{1mm}W_j}\Bigg)W_k,\\
&= W_k + \frac{1}{2\log(M)}\Bigg(W_0+\SUM{j}{k}{l}{\hspace{1mm}W_j}\Bigg)W_k.
\end{align*}
Hence $T_{\mathsmaller{S+k}}$ converges pointwise to $W_{k}$ which converges in distribution to Exp$(1/k)$ as $M\rightarrow\infty$.
\subsection{Proof of Lemma-\ref{lem1}}\label{A0}
Here we show that $\pr(B)\rightarrow 1$ holds in the moderate to high SNR regime (when $\theta = \min\limits_{j\in \widetilde{T}}x_j\gg\sigma$). We choose $\epsilon$ s.t. $\epsilon\gg\sigma$ and $\theta\gg\epsilon$. Now, the knots 
$\tau_{k}, k=1,2,\ldots S$ are independent Rician random variables. Hence,
\begin{equation*}
\pr\Big(\min_{k\in \widetilde{T}}\tau_k\geq\epsilon\Big)= \PROD{k}{1}{S} {\pr\Big(\tau_k\geq\epsilon\Big)}
\geq\PROD{k}{1}{S}{\mathcal{Q}_{1}\Big(\frac{\theta}{\sigma},\frac{\eta}{\sigma}\Big)}
\end{equation*}
Where, $\mathcal{Q}_{1}\Big(\frac{\theta}{\sigma},\frac{\epsilon}{\sigma}\Big)$ is the Marcum $Q$ function, which tends to $1$ as $\frac{\theta}{\epsilon}$ tends to infinity. Hence $\pr\Big(\min_{k\in \widetilde{T}}\tau_k\geq\eta\Big)\rightarrow 1$ for large $\frac{\theta}{\eta}$. Also simultaneously, we note that $\tau_{k}, k=S+1,S+2,\ldots M$ are i.i.d. Rayleigh random variables, hence $\pr\Big(\max_{k\notin \widetilde{T}}\tau_k\leq \eta\Big) = (1-\exp(\frac{-\eta^{2}}{2\sigma^2}))^{M-S}$ which tends to $1$ as $\frac{\eta}{\sigma}\rightarrow\infty$. Hence, $\pr\Big(\max_{k\notin \widetilde{T}}\tau_k\leq \eta\Big)\rightarrow 1$ for large $\frac{\eta}{\sigma}$. So, we can conclude that $\pr(B)\rightarrow 1$ for large $\frac{\theta}{\sigma}$.
\end{document}